\documentclass[conference]{IEEEtran}

\newcommand{\appndexsec}[1]{\subsection{#1}}  
\usepackage{subfigure}
\usepackage{graphicx}   % to use \graphicspath{{../figures/}}
\usepackage{algorithm2e}
\usepackage[x11names]{xcolor}
\usepackage{hyperref}
\hypersetup{colorlinks, linkcolor=purple, urlcolor=blue, citecolor=Green4}

\usepackage{amsthm}   % for proofs
\usepackage{mathtools} % for \begin{align}, etc
\usepackage{thmtools}  % for \begin{restatable}
\usepackage{amssymb}   % for \triangleq
\usepackage{booktabs}    % for tables \toprule
\usepackage{tabulary}    % for tables
\usepackage[nice]{nicefrac}

\newtheorem{definition}{Definition}

\newtheorem{theorem}{Theorem}

\graphicspath{{./figures/}}

\newcommand{\E}{\mathbf{E}}
\newcommand{\argmax}{\operatornamewithlimits{argmax}}
\newcommand{\argmin}{\operatornamewithlimits{argmin}}

\newcommand{\matrx}[1]{\boldsymbol{#1}}
\newcommand{\vectr}[1]{\boldsymbol{#1}}

\newcommand{\calx}{\mathcal{X}}   % set of secrets
\newcommand{\calz}{\mathcal{Z}}   % set of observables
\newcommand{\cali}{\mathcal{I}}   % set of indexes

\newcommand{\cals}{\mathcal{S}}   
\newcommand{\calo}{\mathcal{O}}   
\newcommand{\calc}{\mathcal{C}}   
\newcommand{\calg}{\mathcal{G}}   % grid
\newcommand{\distp}{\mathbf{d}_p} % planar (euclidean ) distance

\newcommand{\mapem}{M_{\textit{em}}}  % the point-to-set map of the em algorithm
\newcommand{\xsolnset}{\Gamma}  % the extended solution set of the em algorithm

\newcommand{\reals}{\mathbb{R}}    % reals
\newcommand{\naturals}{\mathbb{N}} % naturals
\newcommand{\integers}{\mathbb{Z}} % naturals

\newcommand{\call}{\mathcal{L}}   % Lagrange function
\newcommand{\vtheta}{\vectr{\theta}}     % parameter/probability distribution
\newcommand{\vphi}{\vectr{\phi}}           %  parameter/probability distribution
\newcommand{\vq}{\vectr{q}}           %  empirical probability distribution on observables
\newcommand{\va}{\vectr{a}}           % row probability distribution of a mechanism
\newcommand{\vz}{\vectr{z}}        % observed data value
\newcommand{\vv}{\vectr{v}}        % a vector of real numbers
\newcommand{\valpha}{\vectr{\alpha}}   % a curve function
\newcommand{\vzeros}{\vectr{0}}        % a row vector of zeros
\newcommand{\vones}{\vectr{1}}        % a column vector of ones
\newcommand{\mg}{\matrx{G}}    % probability matrix: the conditional probabilities of output at i given that the secret at i is x. 
\newcommand{\invn} {INV-N}                                          
\newcommand{\invp} {INV-P}                                          
\newcommand{\emm} {EM}                                         
\newcommand{\dkl}{D_\textit{KL}} % KL divergence \mathit{KL}                                     

\newcounter{ncomm}

\title{Full Convergence of the Iterative Bayesian Update and Applications to Mechanisms for Privacy Protection}

\author{
\IEEEauthorblockN{Ehab ElSalamouny\IEEEauthorrefmark{1}\IEEEauthorrefmark{2},
                                Catuscia Palamidessi\IEEEauthorrefmark{1}}
\IEEEauthorblockA{  \IEEEauthorrefmark{1}Inria and LIX, \'{E}cole Polytechnique, France \\
                                  \IEEEauthorrefmark{2}Suez Canal University, Egypt}
}

\begin{document}
 
\maketitle
\thispagestyle{plain}  % remove in camera-ready version
\pagestyle{plain}  % remove in camera-ready version

%\runningtitle{Article title}

  %\subtitle{...}

\begin{abstract}
%{Please put abstract here.}
The iterative Bayesian update (IBU) and the matrix inversion (INV) are the main methods to 
retrieve the original distribution from noisy data 
resulting from the application of privacy protection mechanisms. 
We show that the theoretical foundations of the IBU   established in 
the literature are flawed, as they rely on an assumption 
that in general is not satisfied in typical real datasets.
We then fix the theory of the IBU, by providing a general convergence result for 
the underlying Expectation-Maximization method. 
Our framework does not rely on the above assumption,
and also covers a more general local privacy model. 
Finally we evaluate the precision of  the IBU 
on data sanitized with the Geometric, $k$-RR, and \textsc{Rappor} mechanisms,
and we show that it outperforms INV in the first case, while it is comparable to INV in the other two cases. 
\end{abstract}
 %\keywords{Expectation Maximization, local privacy models}

% \journalname{Proceedings on Privacy Enhancing Technologies}
%
%\DOI{Editor to enter DOI}
%  \startpage{1}
%  \received{..}
%  \revised{..}
%  \accepted{..}
%
%  \journalyear{..}
%  \journalvolume{..}
%  \journalissue{..}
% 
%

\section{Introduction}
%With the ever-increasing use of internet-connected devices, such as computers, smart grids, IoT appliances and GPS-enabled equipments,  
%personal data are collected in larger and larger amounts, and then stored and manipulated  
%for the most diverse purposes.  Undeniably, the big-data technology   provides enormous benefits to
%industry, individuals and society, ranging from improving business strategies and boosting  quality of 
%service to enhancing  scientific progress. 
%On the other hand, however, the collection and manipulation 
%of personal data raises alarming privacy issues. 

The problem of developing  methods for privacy protection while preserving utility
%to protect the users' privacy  while preserving the  utility of the data 
has stimulated an active area of research, and several approaches have been proposed. 
Depending on their architecture, these methods can be distinguished in 
 \emph{central} and  \emph{local}~\cite{Dwork:14:Algorithmic}. 
The central model assumes the presence of a trusted  administrator, who has access to the
users’ original data and takes care of the  data sanitization. 
In the local model the users sanitize their data by themselves before they are collected, 
typically by applying some obfuscation mechanism.

The local model is clearly more robust: it  requires no trusted party, 
and it is less vulnerable to security breaches. Indeed, even if a malicious 
entity manages to break into the repository, it will only access  sanitized data.   
On the other hand, guaranteeing a good utility in the local model is more challenging than in the central one. 

Concerning utility, we ought to distinguish two main categories: the \emph{quality of service}, and the \emph{statistical  value} 
of the collected data. 
The former refers to what the user expects from the  service provider, assuming that he has provided voluntarily his data in exchange of some kind of personalized service. 
%Clearly, the more the data are obfuscated, the less the provider can guarantee a good result on the basis of those those data.  
The  statistical utility, on the other hand,  measures the precision of   analyses   on the obfuscated  data   w.r.t.  those  on the original data.

In this paper, we focus on  statistical  utility and on the local privacy model based on  injection of controlled random noise.  
We assume a collection of noisy data produced by a population of users, and we  consider the reconstruction 
of the original distribution, i.e., the distribution determined by the original data. In the privacy literature the main methods that have been proposed 
for this purpose are the matrix inversion technique (INV) \cite{Agrawal:05:ICMD,Kairouz:16:ICML} 
and the iterative Bayesian update (IBU) \cite{Agrawal:01:PDS}, \cite{Agrawal:05:ICMD}\footnote{The IBU was proposed in \cite{Agrawal:01:PDS}, 
and its foundations were also laid in that  paper. However, the name  ``iterative Bayesian update'' was only introduced later on in \cite{Agrawal:05:ICMD}.}.
INV has the advantage of being based on simple linear-algebra  techniques and some post-processing to obtain a distribution. 
The post-processing can be a normalization or a projection on the simplex, and  we will call the corresponding methods \invn{} and \invp{} respectively. 
The IBU is an iterative algorithm that is based on the Expectation-Maximization method well known in statistics, and has the advantage of producing the maximum likelihood estimator (MLE) of the noisy data.

The IBU  has recently  attracted renewed  interest from the community. For example, it has been used in 
\cite{Murakami:18:POPETS,Alvim:18:CSF,Kacem:18:PLAS,Oya:19:SP,Oya:19:THESIS}, relying on the foundations 
established in the seminal paper  \cite{Agrawal:01:PDS}. We have found out, however, that there are various mistakes in the theoretical results of \cite{Agrawal:01:PDS}.
The general problem is that \cite{Agrawal:01:PDS} builds on the convergence result of  Wu \cite{Wu:83:jastat} without paying attention to the assumption that the update process   converges to a point in the ``interior'' of the parameter space (cfr. \cite[Section 2.1]{Wu:83:jastat}). The parameter space, in the case of the IBU, consists of the  distributions candidate 
to be the MLE, i.e., a subset of the $n$-dimensional simplex. The ``interior point'' assumption excludes, therefore, the cases in which the MLE is in the border of the simplex, 
i.e., the non-fully-supported distributions.  Now, this assumption is not necessarily satisfied in the local privacy model, for at least two reasons. 
The first is that the original distribution itself (which has high likelihood to be the MLE) may be not fully supported.  
For example, the Gowalla dataset in San Francisco has   ``holes'', i.e., cells with no checkins, even on a relatively coarsely discretized map,  cfr. Fig.~\ref{fig:holes} (data  available at  
% \verb|gitlab.com/locpriv/ibu/blob/master/gdata/nsf.txt|). 
\url{https://gitlab.com/locpriv/ibu/blob/master/gdata/nsf.txt}).
% {\tt gitlab.com/locpriv/ibu/blob/master/gdata/nsf.txt}). 
This means that the corresponding distribution assigns probability $0$ to those cells. 
The second reason is that, even if the original distribution is fully supported, the MLE may not be. This can happen because some components of the distribution may be underestimated  
if they not contribute much to the likelihood of the observed (noisy) data.

The convergence of the IBU to an MLE in  \cite{Agrawal:01:PDS} (cfr. Theorem 4.3) relies on the MLE being a stationary point, i.e., a point in which the derivatives have value $0$. Indeed, the stationary property would be a consequence of the  ``interior point'' assumption, as proved by  Wu \cite{Wu:83:jastat}. 
But as shown above, that assumption does not hold, and the MLE is not necessarily a stationary point. We show a counterexample in Section~\ref{sec:revisiting_ibu_non_stationary}.

In this paper we fix the foundations of the IBU, and we prove the general convergence of the IBU to an MLE, even in the case in which the ``interior point'' assumption is not satisfied. 
Furthermore, we prove this result in a more general setting than  \cite{Agrawal:01:PDS}. Namely, we assume that each user can apply a different mechanism (also with a possibly different level of privacy), and even change the mechanism several times while the data are being collected. In other words, we assume a \emph{fully local} privacy model. We argue that this is an advantage of our framework, since different users may have different privacy requirements, and even for the same user the  requirements may change over time, or depending on the secret to protect.

\begin{figure}%[h]
\centering 
 %(TV = 0.5759)
      \includegraphics[width=0.4\textwidth]{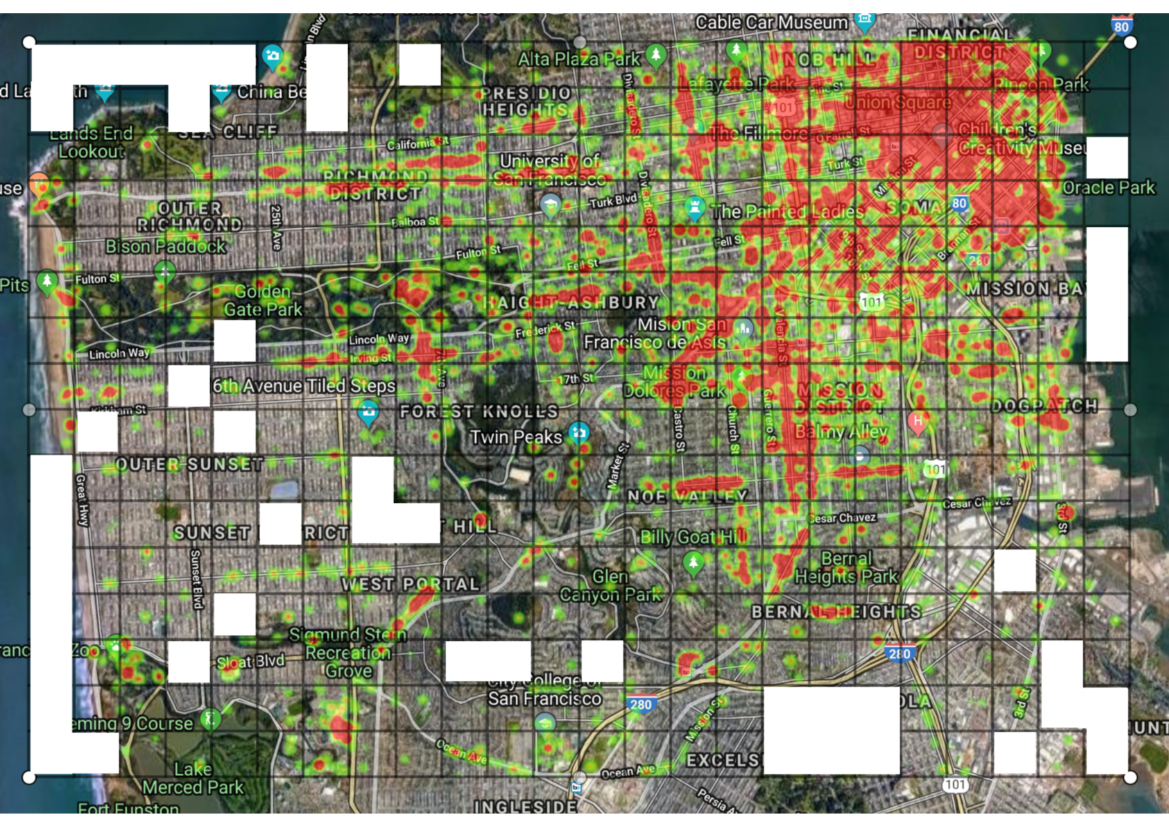}
\caption{Heatmap of the Gowalla checkins in an area in San Francisco. The white polygons represent sets of  cells where there are no checkins.}\label{fig:holes}
\end{figure}

Another problem  in \cite{Agrawal:01:PDS} is the claim that the log-likelihood function has a unique global maximum (cfr. Proposition 4.1), from which \cite{Agrawal:01:PDS} derives that the MLE is unique (cfr. Theorem 4.4). This is not true in general as we show  in Section \ref{sec:revisiting_ibu_non_unique}. As a consequence, Observation 4.1. in  \cite{Agrawal:01:PDS}, stating that as the number of observed data  grows the IBU approximates better and better the original distribution, does not hold either;
we show a counterexample in Section~\ref{sec:not-approximate}. 
%This is not true in general, at least not with some of the modern obfuscation mechanisms, like the  planar geometric and planar Laplace. This does not mean that the INV is better: in those cases, indeed, the INV is not even defined, because the matrix of the mechanism is not invertible. 
Since the whole point of the IBU is to reconstruct  as faithfully as possible the original distribution, it is crucial to ensure that the approximation can be done at an arbitrary level of precision.  As the above counterexample shows, we  get this guarantee  only if the MLE is unique. This motivates us to study the conditions of uniqueness (cfr. Section \ref{sec:unique}).

 %
 % simulated data: binomial and uniform distributions on 100 discrete points on a line
\begin{figure}%[h]
\centering 
 %(TV = 0.5759)
\subfigure[Using \invn{}]{
      \label{fig:geom_binomial_invn}
      \includegraphics[width=0.22\textwidth]{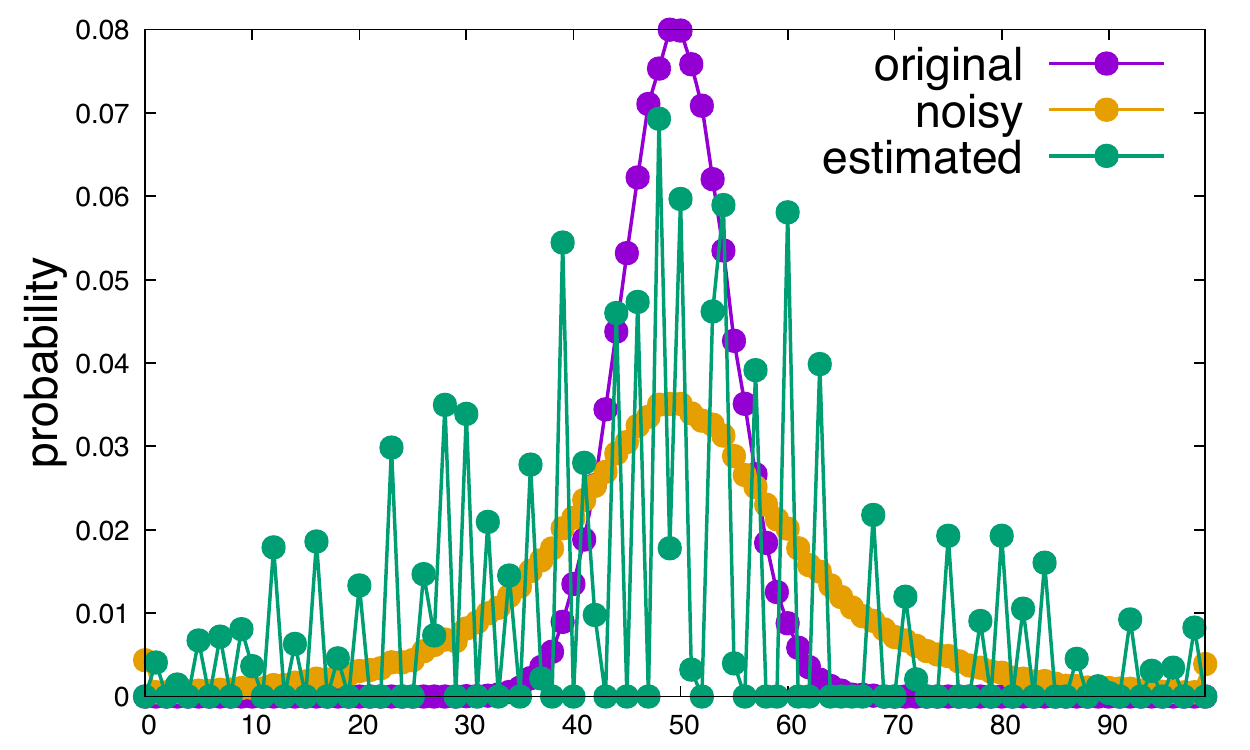}
      }
%(TV=0.5884)      
\subfigure[Using \invn{}]{
      \label{fig:geom_uniform_invn}
      \includegraphics[width=0.22\textwidth]{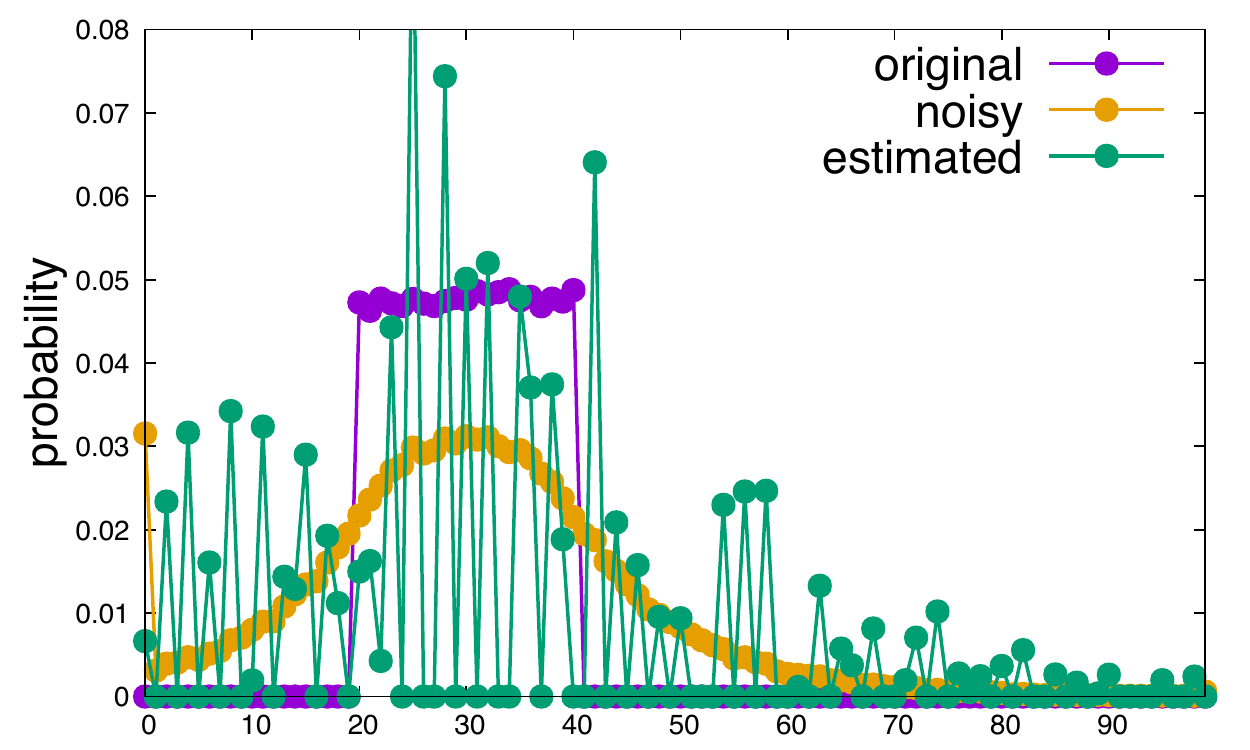}
      }
%(TV= 0.5567)
\subfigure[Using \invp{}]{
      \label{fig:geom_binomial_invp}
      \includegraphics[width=0.22\textwidth]{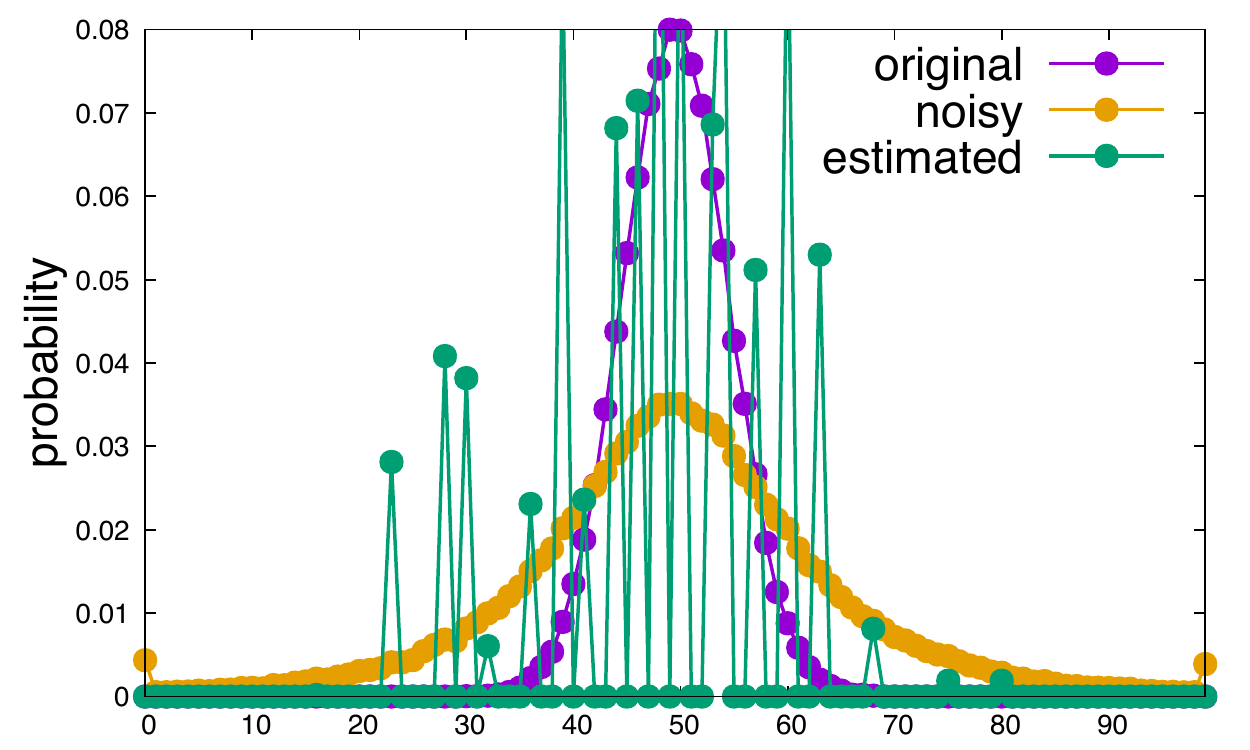}
      }
%(TV=0.6327)      
\subfigure[Using \invp{} ]{
      \label{fig:geom_uniform_invp}
      \includegraphics[width=0.22\textwidth]{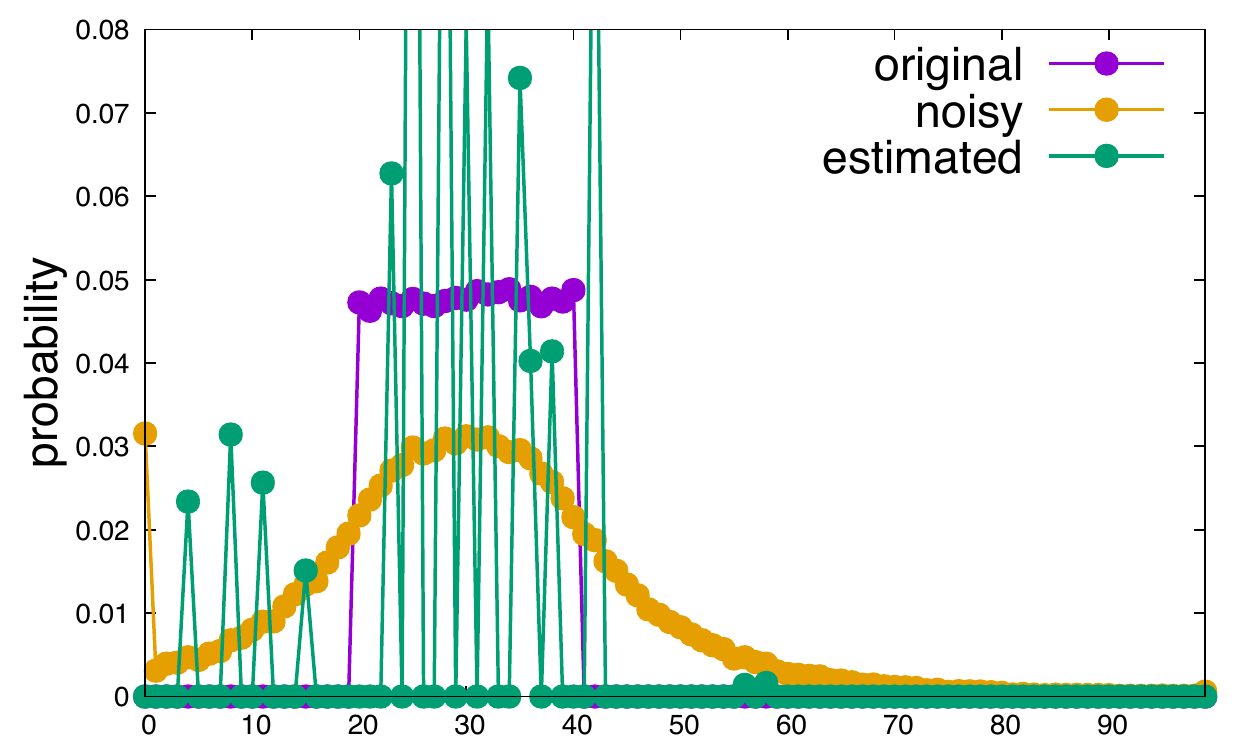}
      }
%(TV = 0.0198)
\subfigure[Using \emm{}]{
      \label{fig:geom_binomial_em}
      \includegraphics[width=0.22\textwidth]{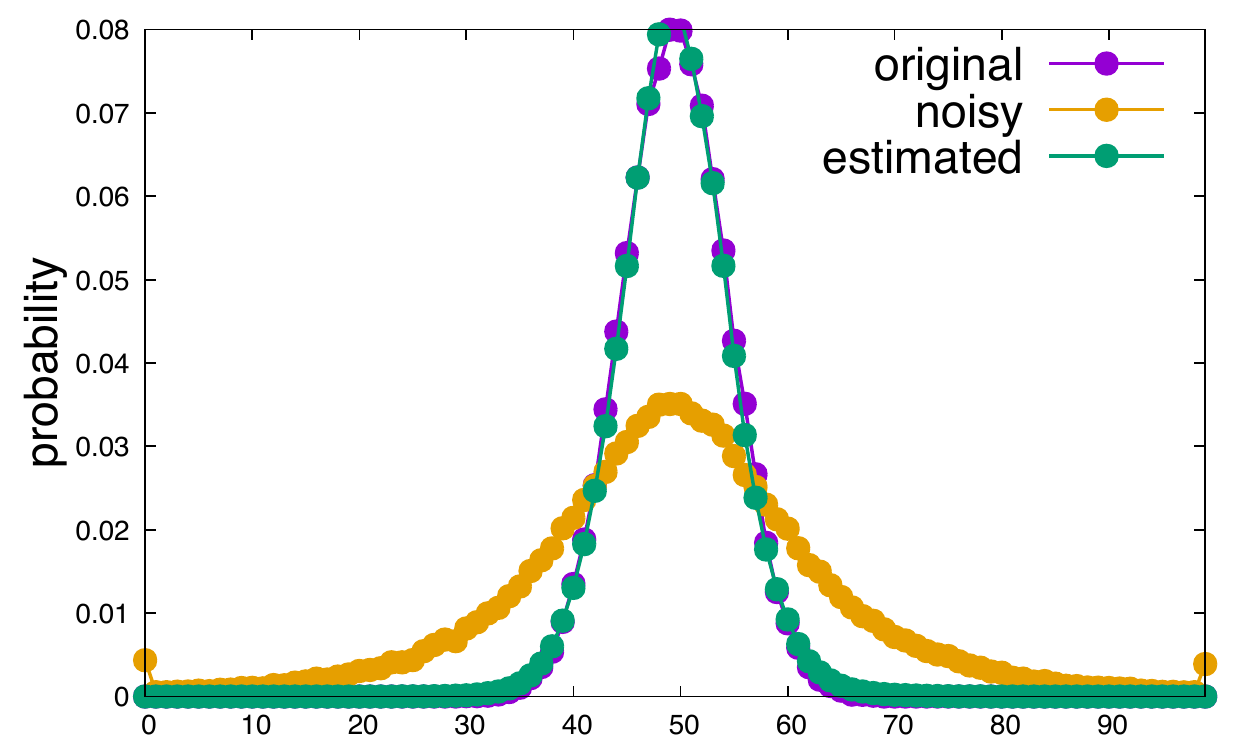}
      }
%(TV = 0.1131)
\subfigure[Using \emm{}]{
      \label{fig:geom_uniform_em}
      \includegraphics[width=0.22\textwidth]{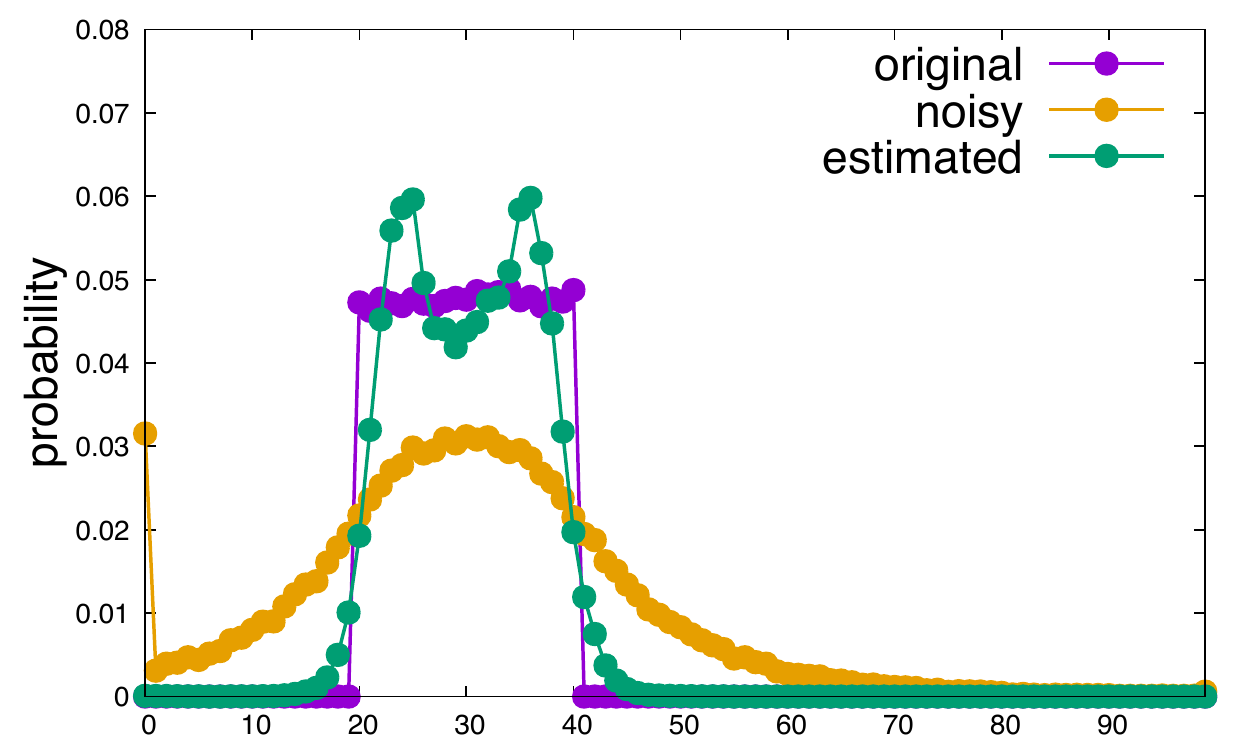}
      }
\caption{Estimating the original distributions from noisy data produced by a geometric mechanism with $\epsilon=0.1$. 
On the left side, the original data are obtained by sampling a  binomial. On the right side, they are sampled from a distribution which is uniform 
on $\{20, \dots,39\}$ and zero outside this range.}
%The estimation in the left side figures is performed using the EM algorithm. The estimation in the right figures is performed 
%using the Matrix Inversion method (INV)}
\label{fig:geometric_on_line}
\end{figure}

Finally, we  compare the performance of the IBU with those of \invn{} and \invp{}. 
 A similar comparison was done in \cite{Agrawal:05:ICMD} (resulting in a favorable verdict for the IBU), but the mechanisms used there were rather different from the modern ones. We  are interested  in comparing their precision  on state-of-the-art mechanisms, in particularly the ones used for \emph{local differential privacy} (LDP) \cite{Duchi:13:FOCS,Kairouz:16:JMLR} and for \emph{differential privacy on metrics} ($d_X$-DP)\cite{Chatzikokolakis:13:PETS}, whose instantiation to geographical distance is  known in the area of location privacy under the name of \emph{geo-indistinguishability} \cite{Andres:13:CCS}.
Two well-known mechanisms for LDP are \textsc{Rappor} \cite{Erlingsson:14:CCS} and $k$-randomized-responses ($k$-RR) \cite{Kairouz:16:JMLR}. 
As for $d_X$-DP and geo-indistinguishability, the typical mechanisms are the Geometric and the Laplace noise. 
We have experimentally verified that (a) the IBU and INV are more or less equivalent for $k$-RR and \textsc{Rappor}, while (b) there is a striking difference when we use the  
geo-indistinguiable mechanisms. Figure~\ref{fig:geometric_on_line} reports our experiments with two different probability distributions on the original data (a binomial distribution and a  distribution uniform on a sub-interval), and applying geometric noise. As we can see, the IBU outperforms by far the INV methods. The experiments are described in detail in Section~\ref{sec:lingeo}.

%The recent works that have considered the IBU have based their results on \cite{Agrawal:01:PDS} and made  similar statements (unique MLE etc.). 
%Given that the condition of  \cite{Wu:83:jastat} is not satisfied by modern mechanisms, in fact, not even by $k$-RR and \textsc{Rappor}, one may wonder whether the IBU converges to a MLE, or even whether it converges at all. Fortunately, it turns out that these results can be recovered. 

\subsection{Contributions}
The contributions of this paper are as follows: 
\begin{itemize}
\item We  show that there are various mistakes in the theory of the IBU as it appears  in the literature. 
\item We fix the foundations of the IBU by providing a general convergence theorem. 
\item The above theorem is actually valid for a generalization of the local privacy model,  in which different users can use different mechanisms, and the same input can generate several outputs, possibly  under different mechanisms.  
\item We identify the conditions under which the MLE is unique. This is important, because only if the MLE is unique then we can 
    approximate the original distribution at an arbitrary level of precision. 
\item We compare  IBU and  INV  on various distributions and mechanisms, showing  that  IBU  outperforms INV in some cases, and is equivalent in the others. 
\end{itemize}

\subsection{Structure of the paper}
Section \ref{sec:preliminaries} presents some preliminaries. 
Section \ref{sec:revisiting} describes the mistakes in the foundations of the IBU. 
In Section \ref{sec:em_conv} we establish a general convergence result of EM algorithms. 
In Section \ref{sec:privacy_model} we
generalize the local privacy model, extend the IBU to work on this model, and also prove the 
convergence of the derived algorithm to the MLEs. 
In Sections \ref{sec:single_inp}, \ref{sec:mul_input} we discuss
 special cases of this algorithm. 
In Section \ref{sec:unique} we study cases in which the MLE is unique. 
In Section \ref{sec:experiments} we experiment the precision of the IBU compared to INV. Section \ref{sec:related_work} describes related work. 
Finally in Section \ref{sec:conclusion} we conclude and 
describe future work.

The proofs of the various results are available in the appendix. The software used for the experiments is available at {\tt gitlab.com/locpriv/ibu}.

%%%%%%%%%%%%
%
%More specifically, we consider two state-of-the-art approaches: 
%\emph{Local differential privacy} (LDP) \cite{Duchi:13:FOCS,Kairouz:16:JMLR} and \emph{differential privacy on metrics} ($d_X$-DP)\cite{Chatzikokolakis:13:PETS}, whose instantiation to geographical distance is particularly popular in the area of location privacy under the name of \emph{geo-indistinguishability} \cite{Andres:13:CCS}.
%Two well-known mechanisms for LDP are \textsc{Rappor} \cite{Erlingsson:14:CCS} and $k$-randomized-responses ($k$-RR) \cite{Kairouz:16:JMLR}. 
%As for $d_X$-DP and geo-indistinguishability, the typical mechanisms are the Geometric and the Laplace noise. 
%
%
%   
%  Trying to understand why the IBU combines so relatively well with geo-indistinguishability, we have analyzed  in detail the foundations of the IBU  established in the seminal paper  \cite{Agrawal:01:PDS}. To our surprise, we have found that there were various problems.  In particular, \cite{Agrawal:01:PDS} claims that as the number of data grows, the result of the IBU approximates the original distribution. This is not true in general, at least not with some of the modern obfuscation mechanisms, like the  planar geometric and planar Laplace. This does not mean that the INV is better: in those cases, indeed, the INV is not even defined, because the matrix of the mechanism is not invertible. 
%

\section{Preliminaries}\label{sec:preliminaries}
%In this section we revise the basic notions and results concerning the maximum-likelihood estimators, the Expectation-Maximization framework, and the iterative Bayesian update. 
\subsection{Maximum-likelihood estimators}\label{sec:ml}  
A statistical model is often used to explain the observed output  of a system. 
Let $\calo$ be the set of potential observables, with generic element $o$, and let $O$ be the random variable ranging on 
$\calo$ representing the output.
Assume that the probability distribution of $O$ depends on a (possibly multi-dimensional) parameter $\vtheta$ taking values in
a space $\calc$. Given an output $o$, the aim is to find the $\vtheta$ that maximizes the
probability of getting $o$, and that therefore is  \emph{the best explanation} of what we have observed.   
To this purpose, it is convenient to introduce the notion of the \emph{log-likelihood function} $L:\calc \to \reals$,
defined as 
\begin{equation}\label{eq:L}
L(\vtheta) = \log P(O =o | \vtheta). 
\end{equation}
where $P(O =o | \vtheta)$ is the conditional probability of $o$ given $\vtheta$. 
The Maximum-Likelihood estimator (MLE) of the unknown parameter is then defined as the $\vtheta$ that maximizes    
$L(.)$ (and therefore $P(O =o | \vtheta)$, since $\log$ is monotone). In some cases these maximizers can be computed by  analytical techniques, but often, 
e.g. when the observables are not dependent directly on the parameter, 
more sophisticated methods are needed.  

\subsection{The Expectation-Maximization framework}\label{sec:em}
The Expectation-Maximization (EM) framework \cite{Dempster:77:jrstat,Wu:83:jastat,McLachlan:08:BOOK} is a powerful 
method for computing the MLE in various statistical scenarios. 
%This framework is used to derive an EM algorithm which iteratively increases the log-likelihood
%function (\ref{eq:L}), given the observed data $o$, and therefore yields the required likelihood maximizer. 
This method is used when the statistical model has hidden data that, if known, would make  the estimation procedure 
 easier. More precisely, let the hidden data be modeled by a random variable $S$ ranging on $\cals$. 
 Then if $S$ is known to have the value $s$, it may be easier 
to work on the ``complete data'' log-likelihood $\log P(S=s, O=o | \vtheta)$ rather than on the (incomplete data) 
log-likelihood in (\ref{eq:L}). Since the value of $S$ is actually unknown,   instead of $\log P(S=s, O=o | \vtheta)$ we consider its   \emph{expected value}, 
computed using a prior approximation $\vtheta'$ of the parameter. 
This expectation yields the function $Q(\vtheta | \vtheta')$ defined as 
\begin{align}\label{eq:Q}
Q(\vtheta | \vtheta') &= \sum_{s \in \cals} P(S = s | O = o; \vtheta') \nonumber\\[-5mm]
                                    &\qquad\qquad\qquad \log P(S = s, O= o | \vtheta).
\end{align}
Let $H(\vtheta | \vtheta')$ be defined as 
\begin{align}\label{eq:H}
H(\vtheta | \vtheta') &= - \sum_{s \in \cals} P(S = s | O = o; \vtheta') \nonumber\\[-5mm]
                                    &\qquad\qquad\qquad  \log P(S= s | O = o ; \vtheta).
\end{align}
It is easy to verify that % (\emph{c.f.} \cite{Wu:83:jastat})  
\begin{equation}   
\label{eq:LQH}
L(\vtheta) = Q(\vtheta | \vtheta') + H(\vtheta | \vtheta'), 
\end{equation}
Using (\ref{eq:LQH}) together with the fact that $H(\vtheta | \vtheta') \geq H(\vtheta' | \vtheta')$ 
(by applying Gibb's inequality) we get the following fundamental property of any EM algorithm: 
%$L(\vtheta) - L(\vtheta') = Q(\vtheta | \vtheta') - Q(\vtheta' | \vtheta') + 
%H(\vtheta | \vtheta') - H(\vtheta' | \vtheta')$. 
%
\begin{equation} \label{eq:dQdL}
L(\vtheta) - L(\vtheta') \geq Q(\vtheta | \vtheta') - Q(\vtheta' | \vtheta')  \quad\forall \vtheta,\vtheta' \in \calc. 
\end{equation}
The above inequality means that if $\vtheta$ is chosen to improve $Q(\vtheta | \vtheta')$ w.r.t. 
$Q(\vtheta' | \vtheta')$, then     $L(\vtheta)$ is also improved w.r.t. $L(\vtheta')$
by at least the same amount. Therefore, $L(\vtheta)$ monotonically grows 
by iterating between two steps: the \emph{expectation} in which $Q(\vtheta | \vtheta')$ is evaluated, and 
the \emph{maximization} which computes a $\vtheta$ that maximizes $Q(\vtheta | \vtheta')$. 
Note that, since $L(.)$ is bounded from above  by $0$,  by the monotone convergence theorem it must converge. 
%\footnote{This follows from the monotone convergence theorem.} 
Table \ref{tab:emnotations} summarizes the above notations and other ones used in Sections \ref{sec:revisiting} and  \ref{sec:em_conv}. 
\begin{table}[t]\caption{Notations used in Sections \ref{sec:preliminaries}--\ref{sec:em_conv}}
%\begin{center}  % used to augment the vertical space between the caption and the table
\begin{tabular}{r p{0.35 \textwidth}}
\toprule
%\multicolumn{2}{c}{\underline{notations of the Expectation Maximization framework}}\\
%\multicolumn{2}{c}{}\\
$\calc$                  & the parameter space\\[0.3ex]
$\vtheta, \vtheta', \vtheta^t$ & elements of $\calc$. $\vtheta^t$ is the parameter 
                                                estimate at time $t$.\\[0.3ex]
$\cals$ & the space of hidden data. \\[0.3ex]
$S$ & random var. representing the hidden data, ranging on $\cals$.\\[0.3ex]
$\calo, \calz$ &   spaces of observed data.\\[0.3ex] 
$O,Z$ & random vars representing  observed data, ranging on $\calo,\calz$ respectively.\\[0.3ex]
$o,z$ & elements of $\calo,\calz$ denoting the observed data.\\[0.3ex] 
$\calx$ & the space of input data\\[0.3ex]
$X, X^i$ & random vars representing the input data, ranging on $\calx$.\\[0.3ex]
$\theta_x$ &probability of $x\in \calx$ based on distribution $\vtheta$.   \\[0.3ex]
$A: \calx \to \calz$ & obfuscation mechanism.   \\[0.3ex]
$a_{x z}$ & probability that $A$ yields  $z$ from  $x$. \\[0.3ex]
$\vq$ & empirical distribution on   $\calz$. \\[0.3ex]
$q_z$ & probability of    $z\in\calz$ according to $\vq$. \\[0.3ex]
$L(\vtheta)$ & the log-likelihood of  $\vtheta$ w.r.t. the observed 
                        data.\\[0.3ex]
$Q(\vtheta | \vtheta')$ & expected complete-data log-likelihood of $\vtheta$.  
%                                       The expectation is taken with respect to the random
%                                       variable $S|O=o,\vtheta'$.
                                       \\[0.3ex]  
$H(\vtheta | \vtheta')$ & difference between $L(\vtheta)$ and $Q(\vtheta | \vtheta')$ (cfr. \eqref{eq:LQH}).\\[0.3ex] 
$\mapem(\vtheta)$  &  the point-to-set map of the EM algorithm, mapping $\vtheta$ to a 
                                     subset of $\calc$.\\[0.3ex]
$I$      & A proper interval in $\reals$. \\[0.3ex]
$\valpha:I \!\to \calc$ &  differentiable curve  from $ I$ to $\calc$.\\[0.3ex]   
$\xsolnset$ &  the solution set of an EM algorithm, $\xsolnset \subset \calc$.\\
\bottomrule
\end{tabular}
%\end{center}
\label{tab:emnotations}
\end{table}
\subsection{Iterative Bayesian Update procedure}\label{sec:ibu}
The iterative Bayesian update (IBU) \cite{Agrawal:01:PDS}
% ehab: the following is to remove confusion about the source of IBU
\footnote{This algorithm was first called `EM-reconstruction' in \cite{Agrawal:01:PDS} and 
afterwards was referred to as IBU in \cite{Agrawal:05:ICMD} and recent publications.}  is an instance of the EM method.
Consider a set of input data represented by the i.i.d. (independent and identically distributed) random variables 
$X^1,X^2,\dots, X^n$, ranging on a space $\calx$ and with distribution $\vtheta$ on $\calx$. Let $\theta_x$ be the probability of $x\in \calx$.    
Suppose that the value of every $X^i$ is independently obfuscated by a given mechanism $A:\calx \to \calz$ to yield a noisy observable $z^i\in\calz$ as shown in Fig.~\ref{fig:local_priv_model_ibu}. 
Given the observations $z^i$ for 
$i=1,2,\dots,n$, the IBU approximates a MLE for them as 
follows. Let $\vq=(q_z: z\in \calz)$ be the empirical distribution on $\calz$, where $q_z$ is 
the number of times $z$ is observed divided by $n$. Let  $a_{x z}$ be the probability that $A$ yields  $z$ from  $x$. The IBU starts with a full-support 
distribution $\vtheta^0$ on $\calx$  (e.g., the uniform distribution), and iteratively produces new distributions by the update rule 
\begin{eqnarray}\label{eq:ibu1}
\theta^{t+1}_x &=& \sum_{z\in \calz} q_z\, \frac {\theta^t_x a_{x z}}{\sum_{u \in \calx} \theta^t_u a_{u z}} \quad \forall x\in \calx.
\end{eqnarray}
%until the difference between $\vtheta^t$ and $\vtheta^{t-1}$ is within a certain margin.
%
\begin{figure}
\center
\includegraphics{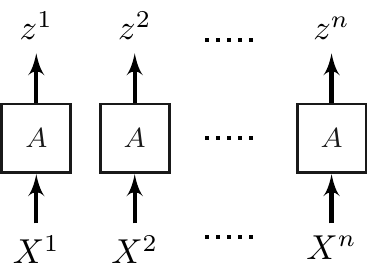}
\caption{The obfuscation model on which IBU is defined. User inputs are $X^1,X^2, \dots, X^n$ taking values in $\calx$. Obfuscating them, by the mechanism $A$,    
yields the noisy data $z^1,z^2, \dots, z^n \in \calz$.}
\label{fig:local_priv_model_ibu}
\end{figure}
The convergence of this algorithm has been studied in \cite{Agrawal:01:PDS},  
 concluding that the IBU converges to a unique MLE which is stationary 
in the space of distributions. In the following section we revisit these claims 
and we prove, by  counterexamples,  that they do not always hold.  
\section{Revisiting the properties of IBU and MLEs}\label{sec:revisiting}
\subsection{The MLE may not be unique}\label{sec:revisiting_ibu_non_unique}
Let $\calx=\{1,2,3\}$, and  consider the obfuscation mechanism $A: \calx \to \calx$ represented by the following stochastic matrix 
(where the inputs are on the rows and the outputs are on the columns. For instance, $a_{11}=    \nicefrac{1}{2}$, $a_{12}=    \nicefrac{1}{3}$, etc.).
\begin{eqnarray}\label{eq:mech}
A &=&\begin{bmatrix}
    \nicefrac{1}{2} &     \nicefrac{1}{3}&     \nicefrac{1}{6}  \\
        \nicefrac{1}{3} &     \nicefrac{1}{3} &     \nicefrac{1}{3}  \\
        \nicefrac{1}{6} &     \nicefrac{1}{3} &     \nicefrac{1}{2}
\end{bmatrix}.
\end{eqnarray} 
Assume  that three users   apply this mechanism and we observe $Z^1=1,Z^2=2,Z^3=3$. 
Given a generic distribution $\vtheta=(\theta_1,\theta_2,\theta_3)$ on $\calx$, 
the log-likelihood is  
$L(\vtheta)=\log P(Z^1= 1, Z^2=2, Z^3 = 3 \mid \vtheta) = \log P(Z^1= 1\mid \vtheta) +\log P(Z^2=2\mid \vtheta) +\log P( Z^3 = 3 \mid \vtheta) =  \log(\nicefrac{\theta_1}{2} + \nicefrac{\theta_2}{3} + \nicefrac{\theta_3}{6}) +\log(\nicefrac{1}{3}) + \log(\nicefrac{\theta_1}{6}+\nicefrac{\theta_2}{3}+\nicefrac{\theta_3}{2})$. 
Fig.~\ref{fig:ml_non_unique_non_convergent_to_real} shows the plot of $L(\vtheta)$, where we  consider only the components $\theta_1, \theta_3$ of  $\vtheta$:  $\theta_2$ is redundant since $\theta_2= 1- \theta_1 - \theta_3$. As we can see, there are infinitely many MLEs, because  every 
$\vtheta$ with $\theta_1=\theta_3$ is a maximum of $L(\vtheta)$, including for instance $(0,1,0)$ and $(\nicefrac{1}{2},0,\nicefrac{1}{2})$. 
This  contradicts   Proposition 4.1 and Theorem 4.4 of \cite{Agrawal:01:PDS} and also   a similar claim in 
% Section 3.2 of
\cite{Murakami:18:POPETS} (Section 3.2) that was based on the above results. The implications of this counterexample lead  
to another refutation as shown in the following.
%: contrarily to what claimed in 
%\cite{Agrawal:01:PDS}, the IBU estimation 
%may not approximate   the real distribution, as shown in the following section. 
% 
\begin{figure}
\includegraphics[width=0.40\textwidth]{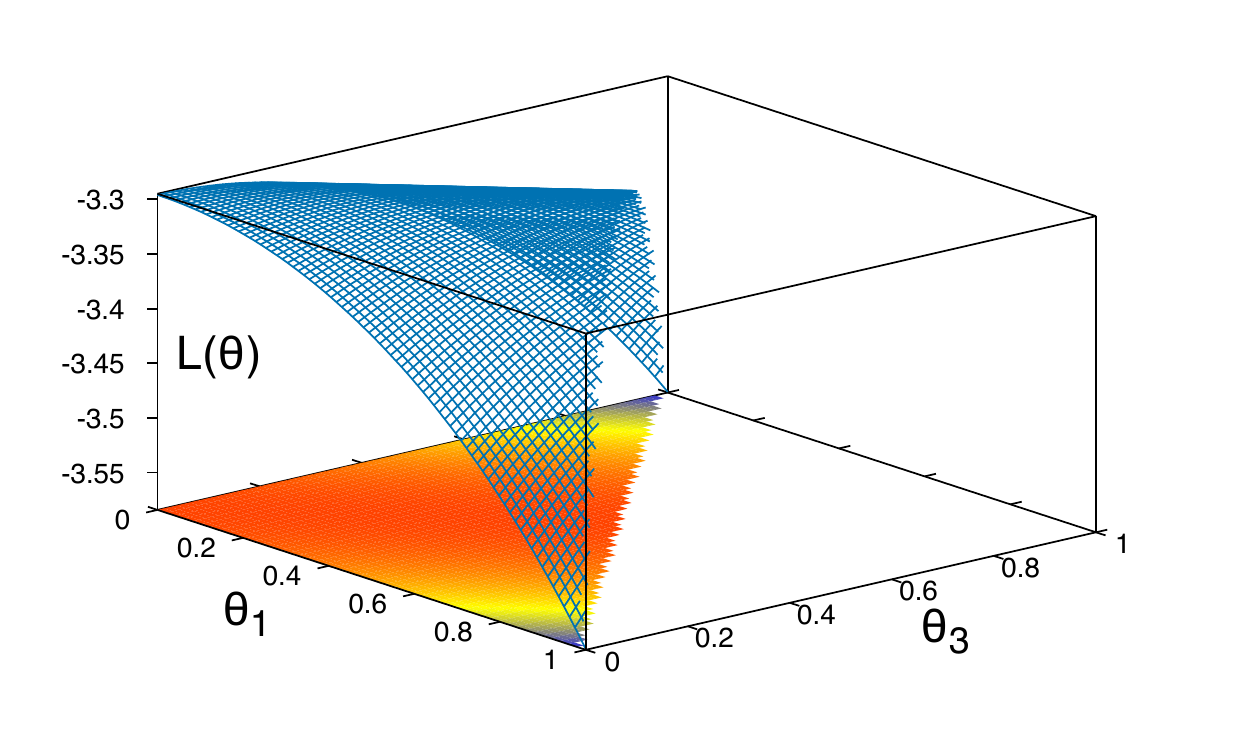}
\caption{The log-likelihood function $L(\vtheta)$ for the distributions
     on $\calx=\{1,2,3\}$ given the mechanism $A$ and the observables $\{1,2,3\}$.}
\label{fig:ml_non_unique_non_convergent_to_real}
\end{figure}
%
%Another example is the following matrix with observations $\{1,2\}$.  
%\[
% A = \begin{bmatrix}
%    2/3 & 1/3  \\
%    1/2 & 1/2  \\
%    1/3 & 2/3 
%\end{bmatrix},
%\]
%
%
\subsection{The IBU may not approximate the true distribution}\label{sec:not-approximate}
Given a mechanism $A$ and a distribution $\vtheta$ on $\calx$, as the number $n$ of input grows, 
the empirical distribution $\vq$ tends probabilistically to the true distribution induced on the output, which is $\vtheta A$.  
From this, \cite{Agrawal:01:PDS} deduces (Observation 4.1) that the result of the IBU approximates the true distribution  $\vtheta $ as $n\rightarrow \infty$.
If $A$ is not invertible, however, there  may be two different  $\vtheta$,  $\vtheta'$ such that 
$\vtheta A = \vtheta' A$, which means that $\vtheta$ and $\vtheta'$ cannot be statistically distinguished 
on the basis of the empirical distribution observed in output,  no matter how large $n$ is. 
The implication for the IBU is that, even in the ``optimal'' case that $\vq=\vtheta A$, 
where $\vtheta$  is the true distribution, the IBU may converge to a different  $\vtheta'$ 
if $\vtheta' A = \vq$ (both $\vtheta$ and $\vtheta'$ are   MLE of $\vq$). 
This contradicts   Observation 4.1 in \cite{Agrawal:01:PDS}. 

As an example, consider again the mechanism  $A$ defined in \eqref{eq:mech} 
and consider the  empirical distribution  $\vq = (\nicefrac{1}{3},\nicefrac{1}{3},\nicefrac{1}{3})$. 
All distributions in the set $\Theta= \{ \vtheta : \theta_1 = \theta_3\}$ satisfy $\vtheta A = \vq$. 
From this it is easy to see that all distributions in $\Theta$  are fixed points of the transformation 
\eqref{eq:ibu1}, i.e.,  $\vtheta^{t+1} = \vtheta^t$ for all  $\vtheta^t \in \Theta$.
This means that the IBU can  converge to any $\vtheta'\in\Theta$ (depending on the starting distribution), 
instead of the true  $\vtheta$. 
For instance, $\vtheta$ could be $(0,1,0)$ and $\vtheta'$ could be  $(\nicefrac{1}{2},0,\nicefrac{1}{2})$. 

\subsection{The MLE may not be stationary}\label{sec:revisiting_ibu_non_stationary}
Theorem 4.3 in \cite{Agrawal:01:PDS} relies on MLE being stationary, i.e., that the derivatives of $L(\vtheta)$ 
have value $0$ on the MLE, and such property is (erroneously) claimed to hold. 
%The assumption of interior MLE leads also in \cite{Wu:83:jastat} to a conclusion that the MLE is stationary. 
The following example shows that this is not true in general. 
Let again $\calx=\{1,2,3\}$, and consider a $3$-RR obfuscation mechanism 
$A': \calx \to \calx$ defined by the following matrix.
\begin{eqnarray}\label{eq:stoch2}
A' &=&\begin{bmatrix}
    \nicefrac{1}{2} &  \nicefrac{1}{4} &\nicefrac{1}{4}   \\
  \nicefrac{1}{4} &  \nicefrac{1}{2} & \nicefrac{1}{4}   \\
    \nicefrac{1}{4}  & \nicefrac{1}{4}  &  \nicefrac{1}{2} 
\end{bmatrix}.
\end{eqnarray}
Assume that four users apply this mechanism and we get the observables $Z^1=1,Z^2=2,Z^3 = 2$ and $Z^4=3$. 
The log-likelihood is $L(\vtheta)=\log(\nicefrac{\theta_1}{2}+\nicefrac{\theta_2}{4}+ \nicefrac{\theta_3}{4}) + 2\log(\nicefrac{\theta_1}{4}+\nicefrac{\theta_2}{2}+ \nicefrac{\theta_3}{4}) + \log(\nicefrac{\theta_1}{4}+\nicefrac{\theta_2}{4}+ \nicefrac{\theta_3}{2}) $.
The plot of  $L(\vtheta)$ (Fig.~\ref{fig:ml_non_interior_non_stationary}) shows that the  unique MLE is $(0,1,0)$.  
On $(0,1,0)$ the partial derivatives are  
$\partial L(\vtheta) /\partial \theta_1 = \partial L(\vtheta) /\partial \theta_3 = - 0.5 $, hence $(0,1,0)$ is not stationary.
We also note that the likelihood surface has no stationary points at all.  
\begin{figure}
\includegraphics[width=0.40\textwidth]{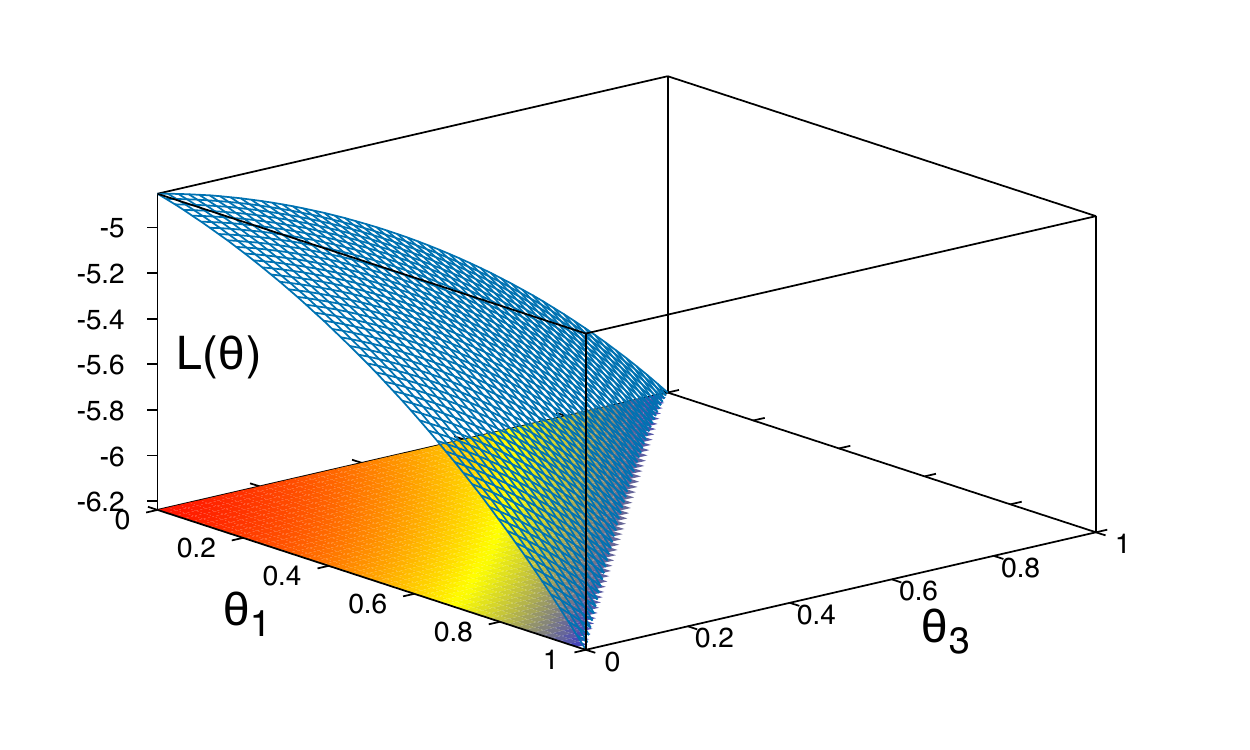}
\caption{The log-likelihood function $L(\vtheta)$ for the distributions 
     on $\calx=\{1,2,3\}$ given the mechanism matrix $A'$ and observables $\{1,2,2,3\}$.}
\label{fig:ml_non_interior_non_stationary}
\end{figure}
\subsection{The IBU may not converge to an interior distribution}
The above misconception in \cite{Agrawal:01:PDS} follows from using the EM properties derived  
in \cite{Wu:83:jastat} while 
ignoring an assumption underlying these properties. That is the limit of the EM updates is interior 
in the parameter space (the probability simplex in our case) (See Section 2.1 in \cite{Wu:83:jastat}). 
The problem is that this assumption may be violated when the IBU is applied to the outcomes of an 
obfuscation mechanism. For example, applying IBU to $A'$ defined in (\ref{eq:stoch2}) and the observations $\{1,2,2,3\}$ yields 
the distribution $(0,1,0)$ which is clearly not interior. Note from Fig.~\ref{fig:ml_non_interior_non_stationary} 
that this distribution is the required MLE which we prove later in this paper (Theorem \ref{thm:alg_conv}).

\section{Full convergence of EM algorithms}\label{sec:em_conv}
%here we study the convergence properties of the em on comact parameter spaces 
In this section we establish a convergence result of the EM algorithms that does not rely on the hypothesis that 
the limit points of the generated sequence $(\vtheta^t)_{t\in\naturals}$ are interior in $\calc$. 
This effort is motivated by the fact that, as discussed in the introduction, such hypothesis is usually not satisfied by typical datasets, e.g. Gowalla, in the contexts of quantitative information flow and privacy.

%\subsection{Extended solution set for EM algorithms}
In the following we will assume, like  \cite{Dempster:77:jrstat,Wu:83:jastat}, that the parameter space $\calc$ 
is a subset of the $m$-dimensional Euclidean space. However we abstract from the assumption made in 
these papers that the limit points are interior to $\calc$, and describe an extended solution set $\xsolnset$ 
to which the EM algorithm converges. We start by defining a curve in $\calc$ as a differentiable function 
$\valpha: I \to \calc$, where $I$ is a proper (i.e. non singleton) interval in $\reals$. Then we can define 
stationary and local maxima for the log-likelihood function along a curve as follows
% The interval I has to be proper, i.e. not containing a single real. 
% The curve has to be in C because a point may be non-stationary 
% nor local maximum on a curve outside C. 

\begin{definition}[Stationary point along a curve]
Consider a parameter space $\calc \subset \reals^m$, and a curve $\valpha: I \to \calc$. 
Let $\vtheta'$ be a point lying on $\valpha$, i.e. $\vtheta' = \valpha(\tau')$ 
for some $\tau' \in I$. Then $\vtheta'$ is stationary for the log-likelihood $L(.)$ along 
$\valpha$ if ${d L(\valpha(\tau))}/{d \tau} =0$ at $\tau'$. 
\end{definition}

\begin{definition}[Local maxima along a curve]
Consider a parameter space $\calc \subset \reals^m$, and a curve $\valpha: I \to \calc$. 
Let $\vtheta'$ be a point lying on $\valpha$, i.e. $\vtheta' = \valpha(\tau')$ 
for some $\tau' \in I$. Then $\vtheta'$ is a local maximizer for the 
log-likelihood $L(.)$ along $\valpha$ if there is $\delta>0$ such that for every $\tau \in I$ satisfying 
$|\tau-\tau'|< \delta$ it holds $L(\valpha(\tau)) \leq L(\valpha(\tau'))$.
\end{definition}

Now we can define our extended solution set $\xsolnset$. 
\begin{definition}[Extended solution set]\label{def:ext_solution_set}
Consider a parameter space $\calc\subset \reals^m$. Then $\xsolnset$ is the set of 
all points $\vtheta'$ such that for each curve $\valpha$ on which $\vtheta'$ lies, either $\vtheta'$ is 
stationary or it a local maximum for the log-likelihood $L(.)$ along $\valpha$.
\end{definition} 
Clearly the extended solution set $\xsolnset$ is larger than the traditional set of stationary points, 
for which the derivatives must be $0$ along every curve. In fact the extended set allows each of its 
members to be,  along each curve, either stationary or a local maximizer (or both). 

We  show that the limits of any sequence 
generated by an EM algorithm must lie in $\xsolnset$, and the log-likelihood   converges to 
a value corresponding to an element in $\xsolnset$. 

% similar to Theorem 2 in \cite{Wu:83:jastat}
\begin{restatable}[EM full convergence]{theorem}{EMconvergence}
\label{thm:conv:em}
Consider a parameter space $\calc \subset \reals^m$, and an EM algorithm in which 
$L(.)$ is continuous on $\calc$, and $Q(\vtheta | \vtheta')$ is continuous with respect to $\vtheta$ and $\vtheta'$. 
Let $(\vtheta^t)_{t\in\naturals}$ be any generated sequence that is contained in a compact subset of $\calc$.   
Then all the limit points of this sequence are in $\xsolnset$ 
and $L(\vtheta^t)$ converges monotonically to $L(\hat\vtheta)$ for some $\hat\vtheta \in \xsolnset$.
\end{restatable}
We remark that Theorem \ref{thm:conv:em} extends Theorem 2 of \cite{Wu:83:jastat} which 
describes the convergence to stationary points based on an assumption that 
the limit points of the EM sequence are interior in the $\calc$.
%the likelihood maximizers 
%are in the interior of the  $\calc$ (assumption 9 in \cite{Wu:83:jastat}), and therefore are stationary. 
Our theorem is stronger since it confirms the 
convergence of the algorithm to the likelihood maximizers whether they are stationary in the interior or lying 
on the border of $\calc$. 
In general there is no guarantee that they are also global maximizers (i.e., MLEs), but if  
$\calc$ is convex and the likelihood function is concave, then they are MLEs. 
\begin{restatable}{theorem}{convex}
\label{thm:convex_space}
Let $\calc$ be a convex parameter space, and $L(.)$ be a concave log-likelihood function over $\calc$. 
Then the extended solution set $\xsolnset$ is exactly the set of global maximizers of $L(.)$.
\end{restatable}
The next section introduces our local privacy model LPM and shows that the conditions of 
Theorem \ref{thm:convex_space} are satisfied for it, so 
that the EM algorithm yields always a global maximizer. 

\section{Local privacy model}\label{sec:privacy_model}
\begin{table}[tbp]\caption{Additional notation used in Sections \ref{sec:privacy_model} -- \ref{sec:unique}} 
%\begin{center}  % used the environment to augment the vertical space between the caption and the table
\begin{tabular}{r p{0.35\textwidth}}
\toprule
$Z^i = (Z^i_j)$        & a multivariate random vector representing the output at index $i$. \\ 
$\vz^i = (z^i_j)$     & the output vector at index $i$ (i.e. the value of $Z^i$), $z^i_j$ is 
                                 the $j$-th observable in the vector $\vz^i$.  \\ 
$k_i$          & the length of the vector $\vz^i$, i.e. the number of observables at index $i$. \\
$A^{i j}$      & the matrix of the mechanism used to obfuscate input $X^i$ to yield the value of $Z^i_j$.\\
$\mg=[g_{x i}]$     & the outputs probability matrix, with $g_{x i}$ being the probability of the observed output at 
                                 index $i$ given that $X^i=x$. \\
%$\hat\calx\subset\calx$   &  a subset containing \emph{likely} elements of $\calx$ 
%                                            (Section \ref{sec:mul_inp_single_mech_arb}).\\
%$\calb$                 &  the convex hull of a given set of points in $\reals^2$ (Section \ref{sec:mul_inp_single_mech_arb}).\\
%$T_{\delta'}(\calb)$         & the region in $\reals^2$ resulting from extending $\calb$ by distance $\delta'$ in 
%                                          all directions (Section \ref{sec:mul_inp_single_mech_arb}).  \\                               
%$\mg(\hat\calx, \cali')$  & a matrix consisting of the rows and columns of $\mg$ that correspond 
%                                        respectively to the elements of $\hat\calx$, and the indices in 
%                                        $\cali'\subset \cali$ (Section \ref{sec:unique}).\\                                  
$\mg(\cali')$  & a matrix consisting only the columns of $\mg$ that correspond 
                                        to the outputs in $\cali'$ (Section \ref{sec:unique}).\\                                  
\bottomrule
\end{tabular}
%\end{center}
\label{tab:notations}
\end{table}
Let $\calx$ be the space of sensitive data of the users. We call the datum of every user an input 
to the system. The inputs are assumed to be independent and drawn from the same probability distribution 
$\vtheta = (\theta_x : \forall x \in \calx)$, hence represented by  i.i.d.
random variables $X^i \sim \vtheta$ where $i=1,2,\dots, n$ and $n$ is the number of inputs. For 
convenience we will denote the set of indices $\{1,2,\dots,n\}$ by $\cali$. 
We assume that every input $X^i$ is obfuscated by $k_i$ mechanisms to yield a multivariate random variable 
$Z^i = (Z^i_1, Z^i_2, \dots, Z^i_{k_i})$, where the elementary variables $Z^i_j$ may have different domains.
We also denote the observed realization of $Z^i$ by $\vz^i = (z^i_1, z^i_2, \dots, z^i_{k_i})$. The mechanisms 
used to produce the elements of $\vz^i$ may vary in their definition and level of privacy. For example if $\calx$ 
is a set of locations, then an input $i$ may yield the output $\vz^i = (z^i_1, z^i_2)$ by running first a Laplace 
mechanism that generates the noisy location $z^i_1$ and then a cloaking mechanism that produces $z^i_2$. 
We assume that every obfuscation mechanism used in the process is known. This allows to compute the probability 
$P(Z^i = \vz^i | X^i=x)$ for every $x \in \calx$. 
\begin{definition}[Outputs probability matrix]\label{def:probability_matrix}
The outputs probability matrix, $\mg=[g_{xi}]$, associates to every 
$x \in \calx$, and every $i \in \cali$ the conditional 
probability $g_{x i}$ of yielding the observed output value at $i$ when 
$x$ is the value of the corresponding input, that is 
$
g_{x i} = P(Z^i = \vz^i | X^i = x). 
$
\end{definition}
Note that $\mg$ is not necessarily square since it consists of $|\calx|$ rows and $n$ 
columns. It is not stochastic either, since its individual columns correspond 
to the indices $\cali$ at which the outputs are drawn using possibly different mechanisms, and from possibly different domains.  
$\mg$ is constructed as follows: 
Let $A^{i j} = [a^{i j}_{x z}]$ be the matrix of the mechanism that was applied by user $i$ to report his 
$j$th observable $z^i_j$. Then
$
g_{x i} = \prod_{j=1}^{k_i} a^{i j}_{x z^i_j}.
$
Table \ref{tab:notations} summarizes the above notations together with others that
will be used later.  Fig.~\ref{fig:local_priv_model} illustrates our local privacy 
model (LPM). 
\begin{figure}
\center
\includegraphics{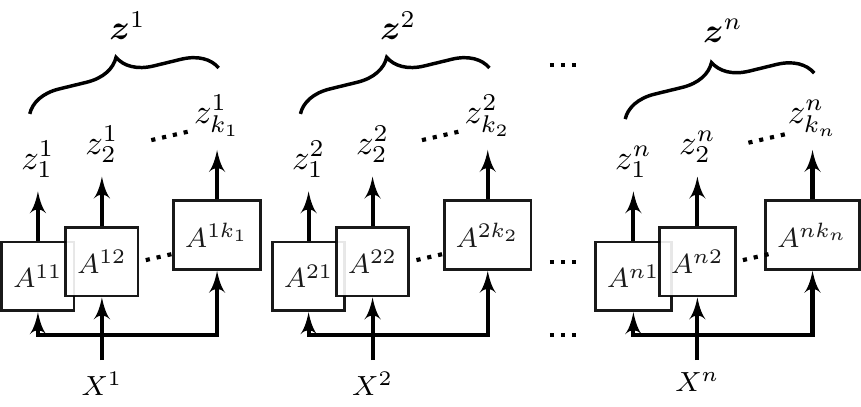}
\caption{The local privacy model. User inputs are $X^1,\dots, X^n$, and corresponding output 
vectors are $\vz^1,\dots,\vz^n$.}
\label{fig:local_priv_model}
\end{figure}
%
%
%
%Table \ref{tab:notations} summarizes the above notations of the local privacy model. 
%%
%\begin{table}[htbp]\caption{Notations of the local privacy model}
%\begin{center}  % used the environment to augment the vertical space between the caption and the table
%\begin{tabular}{r c p{5cm}}
%\toprule
%$\calx$ &  $\triangleq$  &  the domain of the possible values of secrets.\\
%$\vtheta=(\theta_x)$   &   $\triangleq$  &  a probability distribution over $\calx$, with $\theta_x$ being the probability of $x\in \calx$.\\
%$n$      &   $\triangleq$  &  the number of inputs (secrets), also the number of the corresponding outputs. \\
%$X^i$   &   $\triangleq$  &  a random variable taking its value from $\calx$, and representing the hidden input at index $i$, where $i \in \{1,2,\dots,n\}$. \\ 
%$Z^i = (Z^i_j)$      &    $\triangleq$  & a multivariate random vector representing the output at index $i$. \\ 
%$\vz^i = (z^i_j)$   &    $\triangleq$  & the observed value of $Z^i$ (i.e. the output value at index $i$).  \\ 
%$k_i$       &    $\triangleq$   & the length of the vector $\vz^i$, i.e. the number of observables at index $i$. \\
%$A^{i j}$   &    $\triangleq$   & the mechanism (matrix) used to obfuscate input $X^i$ to yield the value of $Z^i_j$.\\
%$g_{x i}$  &    $\triangleq$   & the probability of the observed output at index $i$ given that $X^i=x$. \\
%\bottomrule
%\end{tabular}
%\end{center}
%\label{tab:notations}
%\end{table}

\subsection{The space of distributions and the likelihood function}
In the LPM, we require to find the distribution $\vtheta$ on $\calx$ that maximizes the 
log-likelihood function $L(.)$. Therefore, we define our space of distributions to include 
only those having a finite log-likelihood.
\begin{equation}\label{eq:space_of_distributions}
\calc = \{ \vtheta\in\reals^{|\calx|} : \theta_x\geq 0\,\forall x \in \calx, 
\sum_{x\in \calx} \theta_x =1, L(\vtheta) > -\infty \}. 
\end{equation}
The condition $L(\vtheta)>-\infty$ assures the continuity of $L(.)$ everywhere in $\calc$. 
This is essential for the convergence of the EM algorithm as stated by Theorem \ref{thm:conv:em}. 
The log-likelihood function $L(\vtheta)$ is the logarithm of the joint probability of observed 
output vectors assuming that every input follows the distribution $\vtheta$. 
%That is
%\[
%L(\vtheta) = \log P(\land_{i=1}^n Z^i = \vz^i | \vtheta).
%\] 
Since the outputs at $i=1,2,\dots,n$ are independent, $L(\vtheta)$ can be written as 
%a sum of log-likelihood functions for individual outputs, i.e. 
$L(\vtheta)= \sum_{i=1}^n L_i(\vtheta)$, 
where $L_i(\vtheta)= \log \sum_{x \in \calx} \theta_x\, g_{x i}$ is the likelihood of $\vtheta$ with respect 
to the observed output at $i$. Therefore  
\begin{equation}\label{eq:L_lpm}
L(\vtheta) = \sum_{i=1}^n \log \sum_{x \in \calx} \theta_x\, g_{x i}. 
\end{equation}

A distribution $\hat\vtheta$ is a maximum likelihood estimator (MLE)
%, or likelihood 
%maximizer, 
on $\calx$  if it maximizes the log-likelihood   $L(.)$, i.e. 
\begin{equation}\label{eq:ml}
\hat\vtheta \in \argmax_{\vtheta}  L(\vtheta).
\end{equation}
In general there may be more than one MLE, depending on the  matrix $\mg$. 
Therefore the right side of (\ref{eq:ml}) identifies the set of MLEs rather than 
a single distribution. 
%For the case of a single input (yielding single output vector), we will identify 
%precisely (in Section \ref{sec:single_inp}) these maximizers. For multiple inputs, the problem 
%is harder, and therefore an EM algorithm is used to iteratively compute the a maximizer.

\subsection{Evaluating an MLE using an EM algorithm}

In the following we derive an EM instance algorithm that evaluates a likelihood maximizer $\hat\vtheta$ over $\calx$. 
This algorithm starts with an initial distribution $\vtheta^0$ and then iteratively yields a sequence of estimators $(\vtheta^t)$, 
such that the log-likelihood $L(\vtheta^t)$ monotonically increases with $t$. 
%The algorithm terminates when $L(\vtheta^t)$ 
%reaches a fixed point. 
We  now show how to derive the elements of this sequence. 
The expected complete-data likelihood $Q(.|.)$ can be written as  
\[% \begin{equation}\label{eq:estep}
Q(\vtheta | \vtheta^t) = \E\left[ \log P(\land_{i=1}^n X^i, Z^i| \vtheta) | \land_{i=1}^n Z^i= \vz^i ; \vtheta^t \right].
\]% \end{equation}
\begin{restatable}[E-step]{theorem}{estep}
\label{thm:estep}
The value of $Q(\vtheta | \vtheta^t)$ %defined by (\ref{eq:estep}) 
is given by
\begin{align}\label{eq:der:q}
Q(\vtheta | \vtheta^t) &=\sum_{x\in\calx} \psi_x(\vz, \vtheta^t) \log \theta_x + K(\vtheta^t)
%\nonumber\\[-3mm]
%&\qquad \sum_{x\in\calx} \sum_{i=1}^n P(X^i=x | Z^i= \vz^i ; \vtheta^t)\,\log g_{x i}  
\end{align}
where 
$\psi_x(\vz, \vtheta^t) = \sum_{i=1}^n P(X^i = x | Z^i = \vz^i ; \vtheta^t)$  
and $K(.)$ is a continuous function depending only on $\vtheta^t$.
\end{restatable}
%
% 
%\footnote{
The authors of \cite{Agrawal:01:PDS} studied a special case of our setting, 
that is when the mechanisms are identical and every output vector $\vz^i$ is singleton. 
However their expression of $Q$ in \cite[Theorem 4.1]{Agrawal:01:PDS} is 
not an instance of ours. This is because they define $Q(\vtheta | \vtheta^t)$ 
  as
$\E\left[ \log P(\land_{i=1}^n X^i | \vtheta) | \land_{i=1}^n Z^i= \vz^i ; \vtheta^t \right]$,
which we believe is not consistent with the standard notion of complete-data log-likelihood 
presented in Section \ref{sec:em}.
%} 

Now the maximization step of the algorithm evaluates the 
new update $\vtheta^{t+1}$ to be the maximizer of $Q(. | \vtheta^t)$ with fixed $\vtheta^t$.  
The following theorem characterizes $\vtheta^{t+1}$. 
\begin{restatable}[M-step]{theorem}{mstep}
\label{thm:mstep}
Let $\vtheta^{t+1}$ be the value of $\vtheta$ that maximizes $Q(\vtheta | \vtheta^t)$ in the M-step.  
Then $\vtheta^{t+1}$ is given by
\begin{equation}\label{eq:mstep}
\vtheta^{t+1}_x = 1/n \sum_{i=1}^n P(X^i = x | Z^i = \vz^i ; \vtheta^t) \quad \forall x\in\calx. 
\end{equation}
\end{restatable}
In   light of the above theorem we note that it is important to start the algorithm with a fully-supported distribution 
$\vtheta^0$. In fact, if $\theta^0_x = 0$  \eqref{eq:mstep} 
implies that $\theta_x^t = 0$ for all $t>0$, which means entirely excluding $x$ from the estimation process. 
%Starting 
%with a full-support distribution avoids this undesirable effect. 
%
\subsection{Estimation algorithm}
%Theorem \ref{thm:mstep} can be simply stated as that the new update $\vtheta^{t+1}$ is exactly the average posterior 
%probability distributions over the secrets $\calx$ using $\vtheta^t$ as a prior distribution. 
Using the update statement in Theorem \ref{thm:mstep}, the general EM Algorithm \ref{alg:estimation} 
yields an estimator $\hat\vtheta$ for the hidden distribution over $\calx$. We   show 
in Section \ref{sec:conv_em_lpm} that this estimator is a MLE but may not be unique. 
The input data of this algorithm are the pairs $(z^i_j,A^{ij})$ for every $i,j$ 
where $z^i_j$ is the $j$th observable in the output $\vz^i$, and $A^{i j}$ is the matrix of the mechanism used to yield $z^i_j$. 
These data are then used to evaluate the probabilities $g_{xi}$, which are used to obtain $\vtheta^{t+1}$ based on Theorem 
\ref{thm:mstep}. 
\begin{algorithm}%[h!]
\KwData{ 
$\left(  (z^i_j,A^{ij}=[a^{i j}_{x z}]): i \in \cali , 1 \leq j \leq k_i   \right)$ and $\delta>0$ \;
%$x\in\reals^2$, $\epsilon>0$, $q_\textrm{min} \ge 0$, $Q\subset\reals^2$,
%$w:Q\to\reals^+$
}
\KwResult{Approximate the MLE $\hat\vtheta$ in (\ref{eq:ml}).}
Compute: 
$g_{x i} = P(Z^i=\vz^i | X_i=x) = \prod_{j=1}^{k_i}  a^{i j}_{x z^i_j} \quad \forall i,x:  i \in \cali ,x \in \calx$\;
Set $t=0$ and $\vtheta^0$ to any distribution over $\calx$ with 
%$L(\vtheta^0)>-\infty$ \; 
$\theta^0_x>0$ for every $x \in \calx$\; 
\Repeat{ $| L(\vtheta^t)  - L(\vtheta^{t-1}) | < \delta$ }{
Set $\theta^{t+1}_x = (1/n) \sum_{i=1}^n   \frac {\theta^t_x g_{x i}}{\sum_{u \in \calx} \theta^t_u g_{u i}} $\;
Set $t= t+1$ \;
}
\Return $\vtheta^t$
\vspace{3mm}
\caption{An EM algorithm for estimating the distribution over the hidden data}
\label{alg:estimation}
\end{algorithm}

The condition $\theta^0_x>0$ for all $x \in \calx$ is important not only to involve 
all the elements of $\calx$ in the estimation process as noted earlier, but also to 
guarantee that $L(\vtheta^0)>-\infty$. The latter condition together with the fact 
that $L(\vtheta^t)$ is increasing with $t$, ensures that 
the sequence of estimates $\vtheta^t$ generated by the algorithm is contained 
in a compact subset of $\calc$. 
\begin{restatable} {lemma}{compact}
\label{lemma:compact}
For any $L^-\in\reals$ (i.e.,  $L^->-\infty$), let $\calc' = \{\vtheta \in \calc : L(\vtheta) \geq L^- \}$.
Then $\calc'$ is a compact subset of $\calc$. 
\end{restatable} 
This property meets  the conditions in Theorem \ref{thm:conv:em},   
ensuring the convergence of $\vtheta^t$ to the extended solution set $\Gamma$. 
\subsection{Convergence of the estimation algorithm to a MLE}\label{sec:conv_em_lpm}
The objective of every EM algorithm is to yield the MLE 
for the hidden parameter. In this respect we recall that an EM algorithm in general 
may yield a stationary point which is not necessarily a global maximizer. However 
in the application of EM to our LPM, it turns out that  Algorithm \ref{alg:estimation} 
yields a global maximizer.  

\begin{restatable}{lemma}{convexityLPM}
\label{lemma:convexityLPM}
Let $\calc$ and $L(.)$ be defined as in  \eqref{eq:space_of_distributions} and \eqref{eq:L_lpm} respectively. 
Then $\calc$ is convex and $L(.)$ is concave on $\calc$. 
\end{restatable}

Theorems~\ref{thm:conv:em} and \ref{thm:convex_space}, whose conditions are ensured by  Lemmas \ref{lemma:compact} 
and \ref{lemma:convexityLPM} respectively, allow to  prove  that in our LPM setting  the EM algorithm \ref{alg:estimation} 
converges to a MLE.
\begin{restatable}[Convergence to MLEs]{theorem}{algconv}
\label{thm:alg_conv}
Consider any sequence $(\vtheta^t)_{t \in \naturals}$ of distributions generated by Algorithm 
\ref{alg:estimation}. 
%such that $L(\vtheta^0)>-\infty$. 
Then $L(\vtheta^t)$ converges monotonically 
to $L(\hat\vtheta)$ where $\hat\vtheta$ is a global maximizer of $L(.)$. 
\end{restatable}
It is important to note that $L(.)$ is \emph{not strictly concave} in general, which means 
we may have many MLEs as exemplified earlier in Section \ref{sec:revisiting_ibu_non_unique}.
%We will see  examples of this situation
%in the next two sections, which illustrate instances of our LPM to the cases of single inputs and multiple inputs, respectively. 
%perhaps we should also say that starting from interior distributions is safe since they are in \calc.
%\section{Special cases}\label{sec:cases}
% we skip the case of single input-single output since the following case will provide the same results. 
% Boreale ETAPS
% The topic of the paper is about ML estimation of the data distribution, hence it is not necessary to 
% instantiate the EM algorithm to these elementary cases. So we derive in the following the ML 
% independently of the EM, but we may remark that it can be obtained using EM as well.  
\section{Single input}\label{sec:single_inp}
In the following we consider the instance of our LPM when the system has only one input $X$ which is 
obfuscated multiple times to yield a vector of observables $\vz$. We will consider both the cases when various mechanisms 
are used to produce the elements of $\vz$, and   when only one  is used. 
\subsection {Obfuscation by various mechanisms}\label{sec:single_inp_arb_mech}
Suppose that the input $X$ is obfuscated by $k$ arbitrary mechanisms to produce the observed output  
$\vz = (z_1, z_2, \dots, z_k)$, as shown in Fig.~\ref{fig:local_priv_model_single_inp}.
\begin{figure}
\centering
\includegraphics{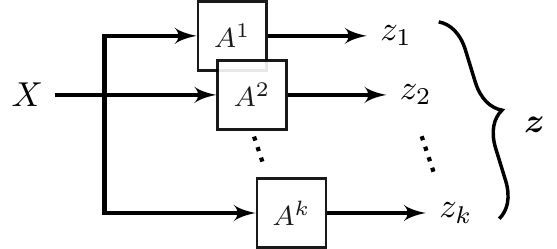}
\caption{The local privacy model for a single input. The user hidden input is  $X$, and corresponding output 
vector is $\vz = (z_1,z_2, \dots,z_k)$.}
\label{fig:local_priv_model_single_inp}
\end{figure}
%
%The likelihood maximizers on $\calx$ 
%depend indeed on the output $\vz$ and the obfuscating mechanisms. 
The following theorem characterizes the MLE in this scenario.  
\begin{restatable}[MLE for a single input]{theorem}{mlsingleinp}
\label{thm:single_inp}
Let $X$ be drawn from a hidden distribution over $\calx$. Suppose 
$X$ is obfuscated to produce an output   $\vz = (z_1, z_2,\dots, z_k)$. 
Then a distribution $\hat\vtheta$ over $\calx$ is an MLE if and only if 
\[
 \sum_{x \in \hat\calx}  \hat{\theta}_x =1\quad\textit{where}
 \quad \hat\calx = \argmax_{x\in\calx} P(\vz | x).
 %\{x: P(\vz | x) = \max_{u \in \calx} P(\vz | u)\}. 
\]  
\end{restatable}
Theorem \ref{thm:single_inp} states that the MLEs are determined by 
the subset $\hat\calx$ consisting of the elements of $\calx$ that maximize $P(\vz | x)$. 
In particular, every distribution that assigns total probability $1$ to the elements of $\hat\calx$ 
is a MLE. 
%Furthermore it is easy to see that the set of MLEs is exactly the set of all convex combinations 
%of these corner distributions. 
%Note that the maximum likelihood 
%is therefore 
%\begin{eqnarray}\label{eq:gx1}
%\max_{x\in\calx} P(\vz | x) &=& \max_{x\in\calx} g_{x 1}
%\end{eqnarray}
%where the $g_{x 1}$ are the elements 
%of the probability matrix (Definition \ref{def:probability_matrix}), which in this case has only one 
%column. 
%
It is therefore clear from Theorem \ref{thm:single_inp} that we may have one or infinitely many 
MLEs depending on $|\hat\calx|$. In fact if $\hat\calx=\{x\}$, i.e. it is 
singleton, then the MLE is unique, namely the one having $P(x)=1$. 
On the other hand, if $|\hat\calx| >1$, then there are infinitely many MLEs. 
%, which are namely all convex combinations of corner distributions at the elements of $\hat\calx$.
%as shown by Theorem \ref{thm:single_inp}. 
This observation is yet another refusal for the claim made in 
\cite[Proposition 4.1]{Agrawal:01:PDS} %, and \cite[Section 3.2]{Murakami:18:POPETS} 
that the log-likelihood function has always a unique MLE.

%examples for cases where the EM converges to non-stationary maximizers: 
%Suppose for example $\calx = \{\textit{yes}, \textit{no} \}$ and consider the simple 
%case where both inputs and observables are drawn from $\calx$. Suppose then that all observables are `yes'. 
%Here the sequence of estimates generated by the EM algorithm converges (correctly) to the limit distribution assigning 
%probability $1$ to `yes' which maximizes the likelihood of the observables. Note here that this limit distribution is not in 
%the interior, and hence the results from \cite{DLR:77:jstat,Wu:83:jastat} are not applicable in this case. 
%%In fact these limit points are non-stationary.  

Finally we recall that every probability $P(\vz | x)$, which is denoted by $g_{x 1}$ (as in Definition \ref{def:probability_matrix}),  
is evaluated from the $k$ obfuscation mechanisms (shown in Fig.~\ref{fig:local_priv_model_single_inp}) 
that produced the elements of $\vz$ . 
%Since in this case we have a single input that is 
%obfuscated by the $k$ mechanisms as in Figure \ref{fig:local_priv_model_single_inp}, 
Let $a^j_{x z}$ be the probability of the mechanism $A^j$ to yield an 
observable $z$ from $x$. Then   
$
g_{x 1} = \prod_{j=1}^k a^j_{x z_j}.
$

\subsection{Obfuscation by identical mechanisms}

Suppose now that the single input $X$ is obfuscated by a fixed mechanism $A$ to yield the 
$k$ observables of the output vector $\vz=(z_1, z_2, \dots, z_k)$. 
In this case the elements of $\vz$ belong all to the same set $\calz$, and therefore we can 
construct an empirical distribution $\vq=(q_z:z\in \calz)$ where $q_z$ is the proportion of $z$ 
in the observed vector $\vz$, i.e. the number of times $z$ appears in $\vz$ divided by 
 $k$. Let $a_{x z}$ be the conditional probability of the mechanism to yield $z$ when its input is $x$. Then 
\begin{eqnarray}\label{eq:gx1_1}
g_{x 1} &=& \prod_{j=1}^k a_{x z_j} \;=\; \prod_{z\in\calz} ( a_{x z})^{k q_z}.
\end{eqnarray} 
From Theorem \ref{thm:single_inp} and \eqref{eq:gx1_1} we then derive
\begin{eqnarray}\label{eq:one}
\hat\calx &=& \argmax_{x\in\calx} g_{x 1} \;=\; \argmax_{x\in \calx} \sum_{z\in \calz} q_z \log a_{x z}. 
\end{eqnarray}
The above equation can   also be interpreted using the Kullback-Leibler divergence $\dkl$ \cite{Cover:06:BOOK}. In fact, 
if we denote by $\va_x$ the conditional distribution 
of the mechanism for $x\in \calx$, we have $\dkl(\vq \,||\, \va_x) = H(\vq) - \sum_{z\in \calz} q_z \log a_{x z}$, 
where $H(\vq) $ is the entropy of $\vq$, which is a constant because $\vq$ is fixed. 
Then, from \eqref{eq:one} we get
\begin{eqnarray}\label{eq:ml-kl}
\hat\calx &=& \argmin_{x\in \calx} \,\dkl(\vq \,||\, \va_x). 
\end{eqnarray}
In other words, the elements  $x\in\hat\calx$ are exactly those for which 
the mechanism distributions $\va_x$ are most similar to the observed empirical 
distribution $\vq$ (w.r.t. $\dkl$). 

\subsection{Relation with the probability of error}

We consider now the impact of the length $k$ of the observed vector $\vz$ on the ML estimation. 
Suppose that the real input value is $\bar x$. Then the observables in the vector $\vz$ are drawn
from $\va_{\bar x}$.  
From \cite[Chapter 11]{Cover:06:BOOK}, we have
\begin{eqnarray}\label{eq:conv_rate}
P\left(\dkl(\vq \,||\, \va_{\bar x}) > \delta\right) &\leq & (1+k)^{|\calz|}\, 2^{-k \delta}.
\end{eqnarray}
This means that the empirical distribution $\vq$ converges exponentially (as 
$k\to \infty$) to $\va_{\bar x}$. If $\va_{\bar x}$ is different from all other rows of 
the mechanism, then by  \eqref{eq:ml-kl} the MLE for a large $\vz$ is unique  and assigns probability 
$1$ to the real input $\bar x$. In other words, as $k\to \infty$ the MLE identifies the real input
with a probability of error $p_e\to 0$. On the other hand, if the row distribution of 
$\bar x$ is not unique, then the MLE is not unique. 
In this case $p_e$ is a constant which depends on the prior distribution on $\calx$. The authors of \cite{Boreale:15:MSCS} 
considered the general case where the mechanism may have identical rows, and 
described upper and lower bounds on $p_e$ that hold for any prior on $\calx$. In particular 
when the mechanism has a unique row for every $x\in \calx$ their upper bound  is equal to $0$, 
which coincides with our observation. 

%The latter observation has been also used by Boreale et al \cite{Boreale:15:MSCS} to obtain a bound on the 
%probability of error at guessing the real input $x$ under the maximum-aposteriori rule (MAP). 
% here we discuss results of Boreale et al,2011

\section{Multiple inputs }\label{sec:mul_input}
In the most general scenario of our LPM, multiple inputs are available, and 
every input $X^i$ is obfuscated repeatedly to produce an observed output vector $\vz^i$ as shown in Fig.~\ref{fig:local_priv_model}. 
%we will refer to every input (and the corresponding output) by its index $i \in \cali$.
In this setting the MLE is obtained using Algorithm \ref{alg:estimation} which applies the update step 
\[ %\begin{equation}\label{eq:update}
\theta^{t+1}_x = (1/n) \sum_{i=1}^n   \frac {\theta^t_x g_{x i}}{\sum_{u \in \calx} \theta^t_u g_{u i}}.
\]
%\end{equation}

The above rule can be instantiated to two special cases that have been considered in the literature \cite{Murakami:18:POPETS,Agrawal:05:ICMD}. 
Both of these works assume that each input is obfuscated once to yield a single observable. However in  \cite{Murakami:18:POPETS} the obfuscation mechanisms 
may be different, while in \cite{Agrawal:05:ICMD} they have to be always the same. We describe these cases in the following.  
\subsection {Obfuscating each input once with different mechanisms}\label{sec:mul_inp_single_mech_arb}
% instantiate the algorithm in \cite{Murakami:18:POPETS}
Suppose that every input $X^i$ is obfuscated once by an arbitrary mechanism $A^i=[a^i_{xz}]$ 
to produce the output $\vz^i$ consisting of a single observable, i.e. $\vz^i=(z^i)$. In this case it is 
easy to see that $g_{x i} = a^i_{x z_i}$, and hence the update step in the general EM algorithm 
\ref{alg:estimation} becomes
\begin{eqnarray}\label{eq:Murakami}
\theta^{t+1}_x &=& (1/n) \sum_{i=1}^n   \frac {\theta^t_x a^i_{x z^i}}{\sum_{u \in \calx} \theta^t_u a^i_{u z^i}}.
\end{eqnarray}
The above formula is equivalent to the update rule in the algorithm proposed by 
\cite{Murakami:18:POPETS}. But we emphasize that it is necessary to 
satisfy our condition on the initial distribution $\vtheta^0$ ($\vtheta^0_x > 0$ for all $x\in \calx$)
to obtain the MLE (see Section~\ref{sec:privacy_model}, in particular the comment after Algorithm~\ref{alg:estimation}). 
The authors of  \cite{Murakami:18:POPETS}, following \cite{Agrawal:05:ICMD}, consider starting with $\vtheta^0=\vq$ (when $\calx=\calz$), 
but this would be a bad choice when $\vq$ is not fully supported. 
\subsection {Obfuscating each input using a single fixed mechanism}\label{sec:mul_inp_single_mech_fixed}
% instantiate the update rule in the algorithm to the IBU in Agrawal paper
% Agrawal 2005 (IBU) and (probably Agrawal 2001)
We now consider the case in which every input $X^i$ is 
obfuscated once by the same mechanism to produce a single observable $z^i\in \calz$. 
In this setting we have  $g_{x i} = a_{x z^i}$, and the update 
step in the EM algorithm \ref{alg:estimation} becomes
\begin{eqnarray}\label{eq:pre-ibu}
\theta^{t+1}_x = (1/n) \sum_{i=1}^n \frac {\theta^t_x a_{x z^i}}{\sum_{u \in \calx} \theta^t_u a_{u z^i}} \quad \forall x\in \calx.
\end{eqnarray}
The above rule can be further simplified using the empirical distribution 
$\vq=(q_z: z\in \calz)$ of the observables. By grouping the 
terms with the same value of $z^i$  we get
\begin{eqnarray*}%\label{eq:ibu}
\theta^{t+1}_x =  \sum_{z\in \calz} q_z\, \frac {\theta^t_x a_{x z}}{\sum_{u \in \calx} \theta^t_u a_{u z}} \quad \forall x\in \calx
\end{eqnarray*}
which is the  IBU 
%\cite{Agrawal:01:PDS} 
described in Section \ref{sec:preliminaries}. 
In the special case when $\calz = \calx$, the authors of \cite{Agrawal:05:ICMD}  propose to start with the 
empirical distribution, i.e.  $\vtheta^0= \vq$, to accelerate the convergence; but as we 
remarked earlier this proposal  is valid only if  $q_z>0$ for every $z$. 

%The above equation is the Iterative Bayesian Update (IBU) \cite{Agrawal:01:PDS} 
%described in Section \ref{sec:preliminaries}. The authors of \cite{Agrawal:05:ICMD}  
%propose to start with the empirical distribution, i.e.  $\vtheta^0= \vq$, but, as 
%% remark the importance of the initial distribution, which was ignored in that paper.
%remarked earlier,  the proposal  is valid only if  \vq has nonzero probability for 
%every $x$.

% may describe the K-rr mechanism and show that $\hat\calx$ consists of observed 
% outputs. Let x be an element that is observed, say at i, and x' be non-observed 
% element. Then g_{x' i} < g_{x i}, and g_{x' i'} = g_{x i'} for all i' \in \cali / i. Thus 
% x' is unlikely compared to x. Note that $\hat\calx$ in this case has no observationally 
% equivalent elements. 

\section{Uniqueness of the MLE}\label{sec:unique}
In this section we investigate the uniqueness of the MLE for a given space $\calx$ of secrets and a probability 
matrix $\mg$ which determines the probabilities of the observed vectors. In the case of a single input, 
Theorem \ref{thm:single_inp} implies that the MLE is unique if and only if 
$\argmax_{x\in \calx} P(\vz | x)$ consists of one element. In the general case of multiple inputs we identify by the following theorem 
sufficient conditions for the uniqueness of the MLE. For any nonempty $\cali' \subseteq \cali$ we define $\mg(\cali')$ to be the 
matrix consisting of the columns of $\mg$ that correspond to the indexes in $\cali'$. Then the required condition of uniqueness
is formulated relative to $\mg(\cali')$ as follows.
\begin{restatable}[Uniqueness]{theorem}{mlunique}
\label{thm:mlunique}
Let $\calx$ be a space on which there is a distribution $\vtheta'$
 with $L(\vtheta')>-\infty$. Assume there is a set of indexes $\cali' \subseteq \cali$ 
such that   $(\vtheta - \vphi) \,\mg(\cali') \neq \vzeros$ for every non-identical distributions $\vtheta, \vphi$ on $\calx$.  
Then there exists a unique MLE on $\calx$ (with respect to the data observed  in the whole  $\cal I$). 
\end{restatable}

Note that for a given distribution $\vtheta$ on $\calx$, the expression $\vtheta \,\mg(\cali')$ describes 
the probabilities of the outputs $\cali'$. Therefore the condition in the above theorem means that 
different distributions $\vtheta, \vphi$ yield different vectors of probabilities for the outputs $\cali'$. 

Next, we revise some well known mechanisms from the privacy literature, and investigate the uniqueness of the MLE 
when they are applied in the setting of IBU (Fig.~\ref{fig:local_priv_model_ibu}). We show that if the space of secrets of interest  $\calx$ 
is defined to be the same as the set of observed values, then there is a unique MLE on $\calx$.  

\subsection{k-RR mechanisms}\label{sec:uniqueness:krr}
The $k$-ary randomized response mechanism, $k$-RR, was originally introduced by Warner \cite{Warner:65:jastat} 
to sanitize sensitive data that are drawn from a binary alphabet ($k=2$). 
%for the privacy protection when the alphabet of the sensitive data is binary ($k=2$). 
Then it was extended by \cite{Kairouz:16:JMLR} to arbitrary $k$-size alphabets. 
This mechanism applied to an element $y$ of the alphabet produces another element $z$ with  probability:
\begin{equation}\label{eq:krr}
P(z|y) = \frac{1}{k-1+ e^\epsilon} 
\left\{ 
\begin{array}{l l}
e^\epsilon & \mbox{if $z=y$}\\
1                & \mbox{if $z\neq y$} 
\end{array} 
\right.
\end{equation} 
It is known that $k$-RR satisfies $\epsilon$-local differential privacy, and \cite{Kairouz:16:JMLR} has proved it to be 
optimal (under the LDP constraints) for a range of statistical utility metrics, e.g., Total Variation and KL divergence. 
Now, using Theorem \ref{thm:mlunique},  
we show that if the space of secrets of interest $\calx$ is defined to be the same as the set of the observed values 
%we show that if $\calx$ is defined to be consisting of observed values, 
then there is a unique MLE on $\calx$. Note that the mechanism may be defined on a superset of $\calx$. 

\begin{restatable}{corollary}{mluniquekrr}
\label{thm:mlunique-krr}
Let $\calx$ be the set of values reported by a $k$-RR mechanism. Then there is a unique MLE on $\calx$. 
\end{restatable}

\subsection{Geometric mechanisms}\label{sec:geometric}
The (linear) geometric mechanism probabilistically maps the space of integers $\integers$ to itself. 
Precisely, given the parameter  $\epsilon>0$, it maps every $y \in \integers$ to $z\in \integers$ with probability 
\begin{equation}\label{eq:geometric}
P(z |y) = c \, e^{-\epsilon |z-y|}, \quad \mbox{where}\; \;   c=\frac{1-e^{-\epsilon}}{1+e^{-\epsilon}}.
\end{equation}
The geometric mechanism   is known to be 
$\epsilon$-differential private, and furthermore is universally optimal  \cite{Ghosh:09:STOC}.   
Now, using Theorem \ref{thm:mlunique}, we can show that again if we define $\calx$ to be the set 
of values reported by the mechanism, then there is a unique MLE on $\calx$.   
\begin{restatable}{corollary}{mluniquegeometric}
\label{thm:mlunique-geometric}
Let $\calx$ be the set of values reported by a geometric mechanism. Then there is a unique MLE on $\calx$. 
\end{restatable} 
One well known variant of the geometric mechanism is its \emph{truncated} version \cite{Ghosh:09:STOC}  which works on a bounded range 
of integers. Let $r_1, r_2$ be two integers with 
$r_1< r_2$. Then the mechanism maps every $y$  between $r_1$ and $r_2$  into an integer $z$ in the same range 
with probability 
\begin{align}\label{eq:tgeometric}
P(z |y) &= c_z \, e^{-\epsilon |z-y|},\\
\mbox{where}\; \; c_z &= \frac{1}{1+e^{-\epsilon}} 
\left\{ 
 \begin{array}{l l}
1            &\mbox{if $z \in \{r_1, r_2\}$} \\
1-e^{-\epsilon} &\mbox{if $r_1< z< r_2$} \\
0            & \mbox{otherwise}
\end{array} \label{eq:tgeom-cz}
\right..
\end{align}
Also in this case we can show that there is a unique MLE on the set $\calx$  of observed values. Again, $\calx$ may not contain all the integers between $r_1$ and $r_2$. 

\begin{restatable}{corollary}{mluniquetgeometric}
\label{thm:mlunique-tgeometric}
Let $\calx$ be the set of values reported by a truncated geometric mechanism. Then there is a unique MLE on $\calx$. 
\end{restatable} 

\section{Experimental evaluation}\label{sec:experiments}
In the following we experimentally evaluate the performance of the EM algorithm (\emm{}) by measuring the statistical 
distance between the original and estimated distributions. We focus on the setting of the IBU,  where every input 
 $x\in \calx$ generates a single observable $z \in \calz$ using a mechanism defined by a stochastic matrix $A$. 
We will also compare the IBU with the matrix inversion technique \cite{Agrawal:05:ICMD,Kairouz:16:ICML}, where
the empirical distribution $\vq$ on $\calz$ is used to estimate the original distribution by evaluating the vector 
$\vv=\vq \, A^{-1}$ (where $A^{-1}$ is the inverse of $A$), and then transforming $\vv$ into a valid distribution on $\calx$. 
This transformation may be done by truncating the negative components of $\vv$ to $0$ and then normalizing the vector;  
or alternatively by projecting $\vv$ onto the  simplex in $\reals^{|\calx|}$ using e.g. the algorithm in 
\cite{Wang2013ProjectionOT}. We refer to these two methods as \invn{} and \invp{} respectively. 

In our experiments we consider two classes of obfuscation mechanisms. 
The first class includes the mechanisms that satisfy $\epsilon$-geo-indistinguishablity \cite{Andres:13:CCS}, 
namely the geometric, Laplace, and exponential mechanisms \cite{Chatzikokolakis:17:POPETS}.
The second category includes the mechanisms that satisfy $\epsilon$-local differential privacy, namely the $k$-RR and 
the Google's \textsc{Rappor} mechanisms. We will use two types of data:  synthetic in the linear space, and real-world 
geographic data from the Gowalla dataset \cite{Gowalla:190518} in the planar space. 
\subsection{Synthetic data in the linear space}\label{sec:lingeo}
We define the space of input data $\calx$ to be the numbers $\{0,1,\dots,99\}$. 
We assume that the users obfuscate their data 
using a truncated geometric mechanism (cfr. Section~\ref{sec:geometric})  with a
strong geo-indistinguishability level ($\epsilon=0.1$). Then we apply the  methods to estimate 
the original distribution. We run two experiments. In the first one we draw $10^5$ inputs from a binomial 
distribution with parameter $p=0.5$. 
%It is well known that the binomial distribution converges to the gaussian distribution when the data size is large. In our setting we sample $10^5$ data inputs independently. 
In the second experiment, we sample the same number of inputs from a distribution that is uniform on the 
elements $\{20,\dots,39\}$ and assigning probability $0$ to every other element in $\calx$. 
%\footnote{Note that the real distribution on the user data is determined by the frequencies 
%of individual values, not the theoretical distribution that we use in these experiments to draw the values.} 
In the two experiments, 
the data are obfuscated by the above truncated geometric mechanism, 
and then the methods \invn{}, \invp{}, and \emm{} are applied on the resulting noisy data to estimate the original distribution. 
The results of these   experiments are shown in Fig.~\ref{fig:geometric_on_line}. Clearly the \emm{} outperforms 
the other two methods in approximating the real distributions, despite some distortion  in 
the case of the uniform distribution  (Fig.~\ref{fig:geom_uniform_em}) due to the discontinuities of this distribution. 
%In fact this distortion is reduced with larger data samples. 
Note that the distribution estimated 
by \invp{} is slightly more similar to the true distribution than the one produced by \invn{}. 

\subsection{Geographic data from Gowalla, planar geometric noise}\label{sec:Gowalla}
We consider now the case in which the elements of $\calx$ are locations in the planar space. 
%For the sake 
%of privacy we assume the users obfuscate their locations before reporting them to the data collector, and from the reported 
%noisy data we require to estimate the real distribution of the users locations. 
We consider a zone in the North part of San 
Francisco  bounded by the latitudes $37.7228, 37.7946$, and the longitudes $-122.5153, -122.3789$, covering 
 an area of $12$km $\times$ $8$km. We then discretize 
the region into a grid of $24\times 16$ cells of  width  0.5km  (see Fig.~\ref{fig:sf_map}). 
We approximate every location by the center of 
its enclosing cell, and  define the space $\calx$ to be the set of these centers. 
Now we use the check-ins in the Gowalla dataset restricted to this region ($123273$ check-ins) as the real users data. 
We obfuscate these data using the truncated planar geometric mechanism \cite{Chatzikokolakis:17:POPETS} 
(described also in Section \ref{sec:plan_geom} of the appendix),  known to satisfy 
$\epsilon$-geo-indistinguishability. This mechanism is defined by a formula similar to \eqref{eq:geometric}, except that $x$ and $z$ are location on a plane and $|z-x|$ is replaced by the Euclidean distance between $z$ and $x$. 
We apply the mechanism with $\epsilon=1.0$ 
to produce the noisy data. 
%\footnote{Note this this value introduces less noise compared to the value 
%$0.1$ which we used in the linear case. This is because here $|\calx|=384$ cells which is larger than the $\calx=100$ in 
%the linear case, while we have a similar number of data samples. This makes the learning the original distribution relatively reliable.}. 
 %due to  the sophisticated projection onto the probability simplex.  
%
\begin{figure}[t]
\centering
\includegraphics[width=0.4\textwidth]{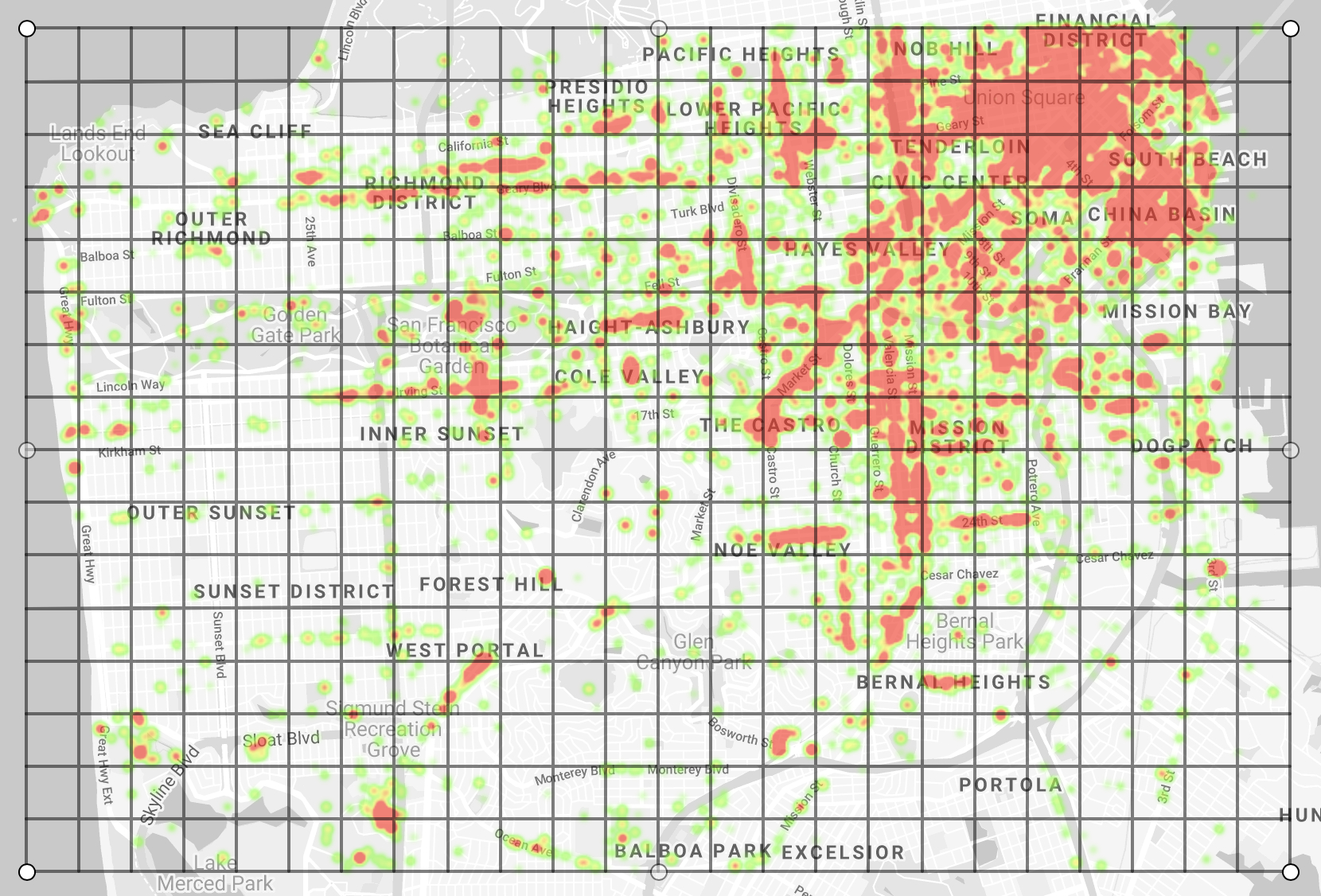}
\caption{The north part of San Francisco, discretized in a grid of $24\times16$ cells, covering an area of 12km$\times$8km. The   
heatmap is based on the check-ins  from the Gowalla dataset.}
\label{fig:sf_map}
\end{figure}
 \begin{figure}[t]
\centering 
\subfigure[\text{Original}]{
      \label{fig:original}
      \includegraphics[width=0.22\textwidth]{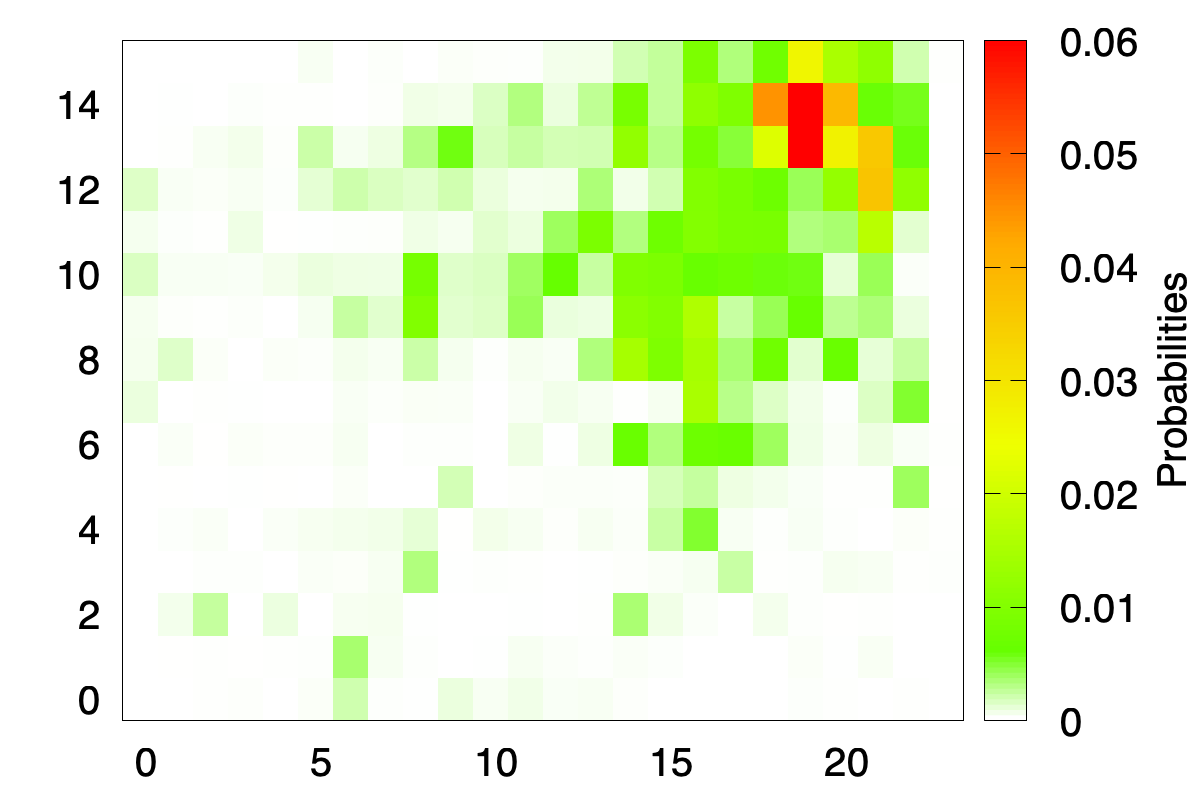}
      }
%\hspace{0pt}
\subfigure[Noisy  (TV = 0.4544)]{
      \label{fig:geom_grid_noisy}
      \includegraphics[width=0.22\textwidth]{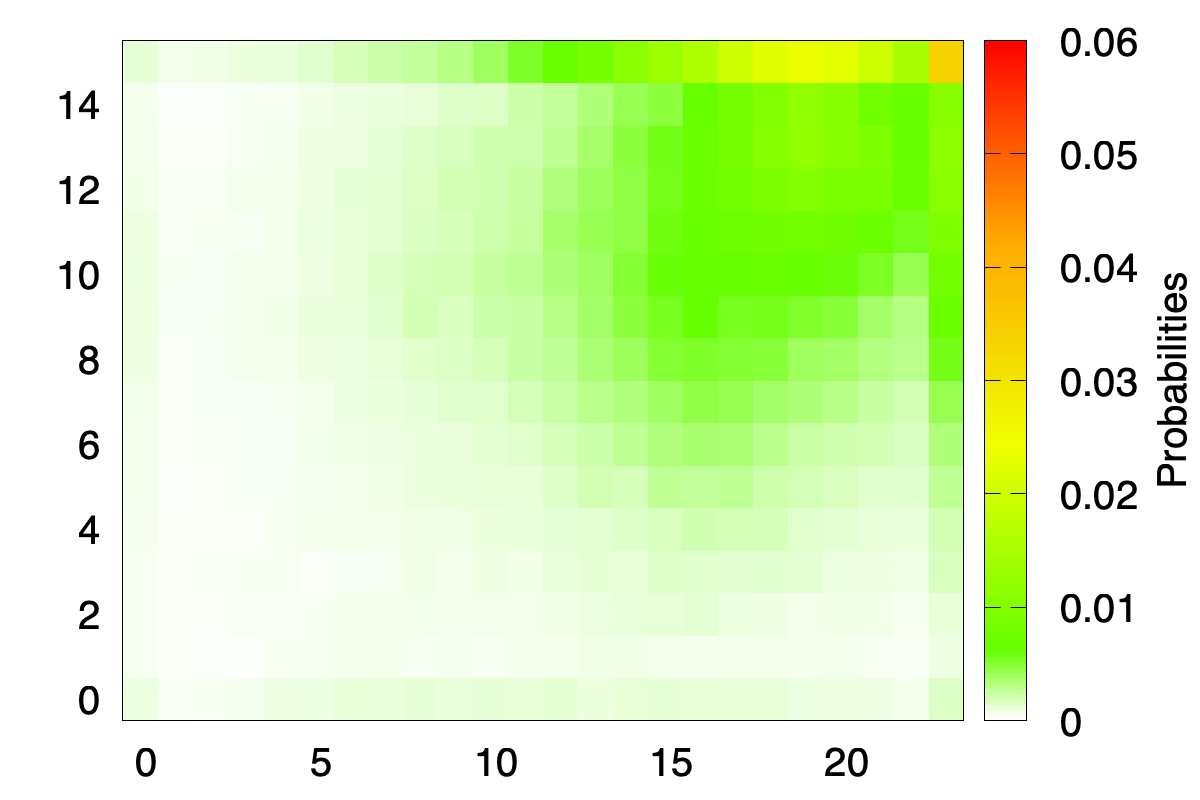}
      }
\subfigure[\invp{} (TV = 0.4567)]{
      \label{fig:geom_grid_estimated_invp}
      \includegraphics[width=0.22\textwidth]{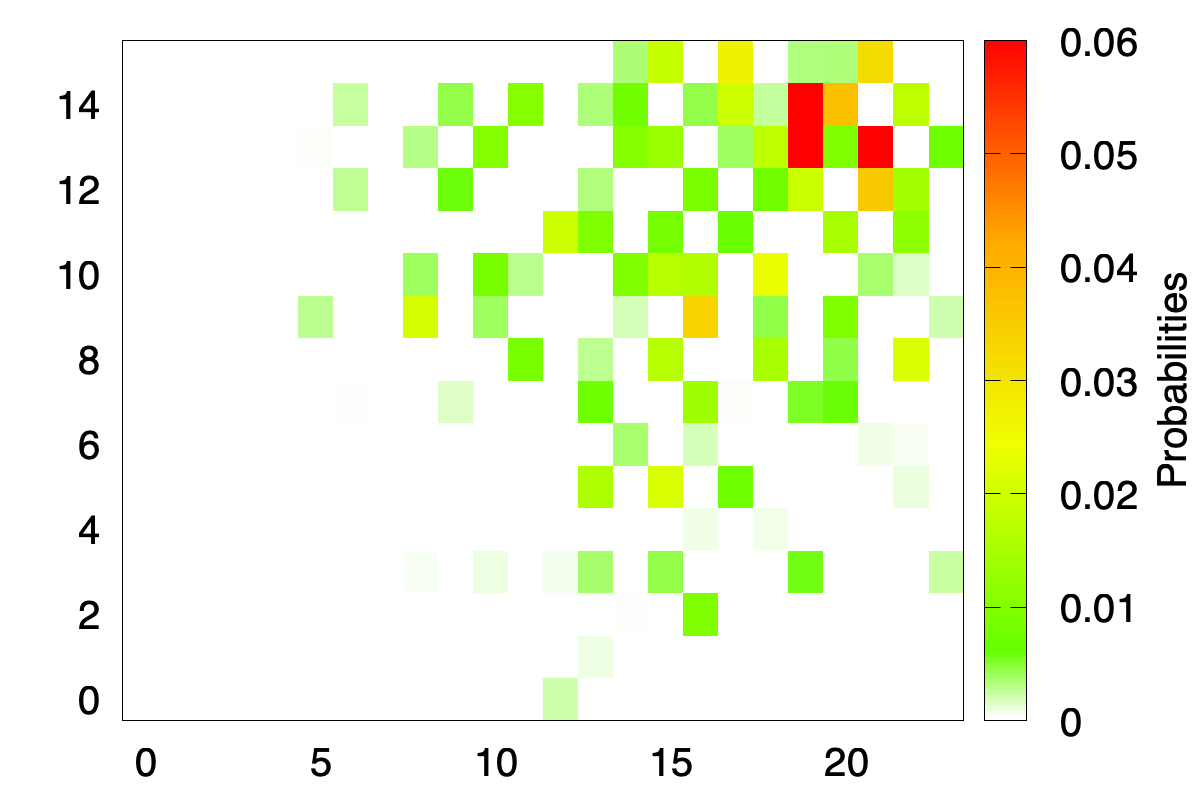}
      }
\subfigure[\invn{} (TV = 0.4633)]{
      \label{fig:geom_grid_estimated_invn}
      \includegraphics[width=0.22\textwidth]{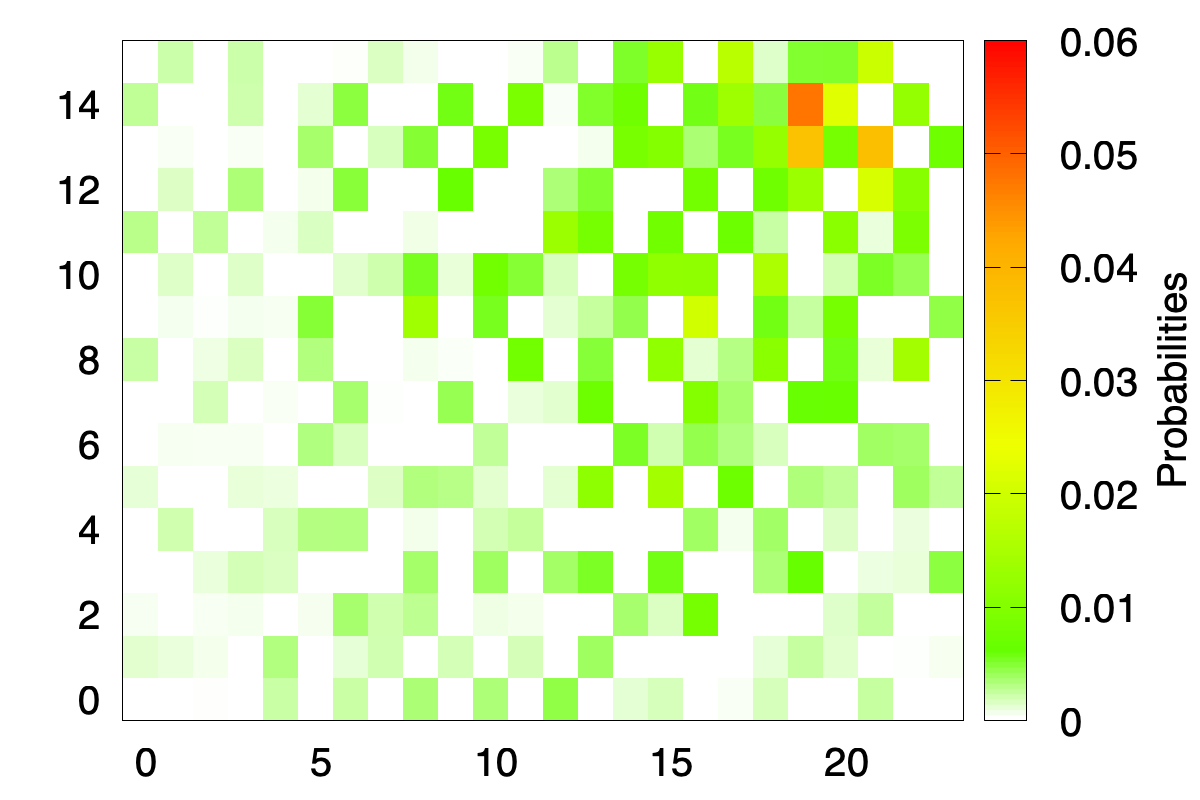}
      }
\subfigure[\emm{} (TV = 0.2238)]{
      \label{fig:geom_grid_estimated_em}
      \includegraphics[width=0.22\textwidth]{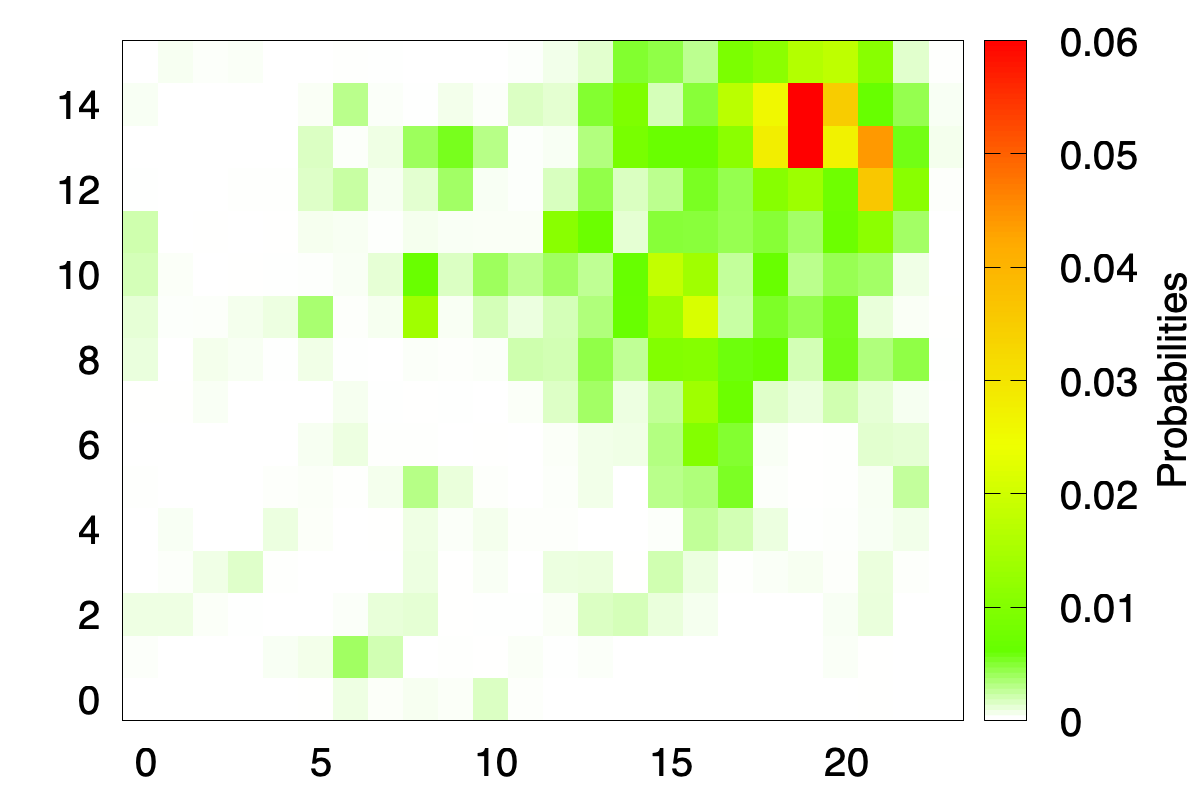}
      }
\caption{Estimation of the Gowalla  original distribution. The noisy data are 
produced by a truncated planar geometric mechanism with $\epsilon=1.0$. 
The shown TV is the total variation distance between the estimation and  the original distribution.}
\label{fig:geometric_on_grid}
\end{figure}
Fig.~\ref{fig:geometric_on_grid} shows the original and noisy distributions on the grid and   the 
distributions estimated by \invn{}, \invp{} and \emm{}. Also in this   case  
we observe that the EM method outperforms both \invn{} and \invp{}. We   also observe that \invp{} is slightly better than \invn{}.
Finally note that the EM in this experiment yields an MLE that is not interior (in the probability simplex) 
since many cells have probability $0$ as shown in Fig.~\ref{fig:geom_grid_estimated_em}. This is a practical situation 
that violates the assumption in \cite{Wu:83:jastat}, and hence motivates our revision of the IBU foundations. 

\subsection{Estimation under various levels of privacy}
We perform again the  experiment described in  Section~\ref{sec:Gowalla}, but now we  vary the privacy parameter $\epsilon$. We use a range for $\epsilon$ between $0.2$ and $6.0$, and for every value we 
run the obfuscation-estimation procedure $100$ times to obtain the boxplot in Fig.~\ref{fig:box_estimators_geom}. The estimation quality is measured  in terms of the 
distance from the true distribution (defined   by the frequencies 
of the check-ins in Gowalla). We consider two  statistical distances: 
the Total Variation (TV) and the Earth Mover's Distance (EMD)\footnote{The EMD is the minimum cost to transform one distribution into another one by moving probability masses between   cells.  \cite{Alvim:18:CSF} argues that it is  particularly appropriate   for location privacy}.
As expected, the estimation quality of all methods improves with larger values of $\epsilon$, corresponding to introducing less  noise. 
We observe that  \emm{} outperforms  \invn{} and \invp{}, especially at low values of $\epsilon$ (stronger levels of privacy). 
We also remark that the EMD is more indicative than the TV for comparing  distributions on locations: Consider the
 distributions  resulting from \invp{} and \invn{}  in  Fig.~\ref{fig:geometric_on_grid} (where $\epsilon=1$): arguably, the first is more similar to the true distribution than the second one, the latter being  quite scattered away from the locations where the probability is accumulated. 
 This qualitative superiority of  \invp{} over  \invn{}   is reflected by the EMD (Fig.~\ref{fig:box_geom_emd}), where \invp{} is 
significantly better than \invn{} for  $\epsilon=1$. On the other hand, 
in Fig.~\ref{fig:box_geom_tv}, the TV distances for \invn{} and   \invp{}  for  $\epsilon=1$ on are almost identical (although they differ from the original for opposite reasons: the one from  \invp{} is too skewed and the one from  \invn{} is too scattered). 
%
% in which the TV distances for \invn{} and \invp{} (and also that of the noisy distribution) are almost identical, while the distribution 
%resulting from \invp{} is more skewed to the original distribution making the EMD smaller compared to that of \invn{}.  
%This similarity is reflected by the EMD metric in Figure \ref{fig:estimators_geom}.  
%
\begin{figure}[t]
\centering 
\subfigure[Total variation distance]{
      \label{fig:box_geom_tv}
      \includegraphics[width=0.4\textwidth]{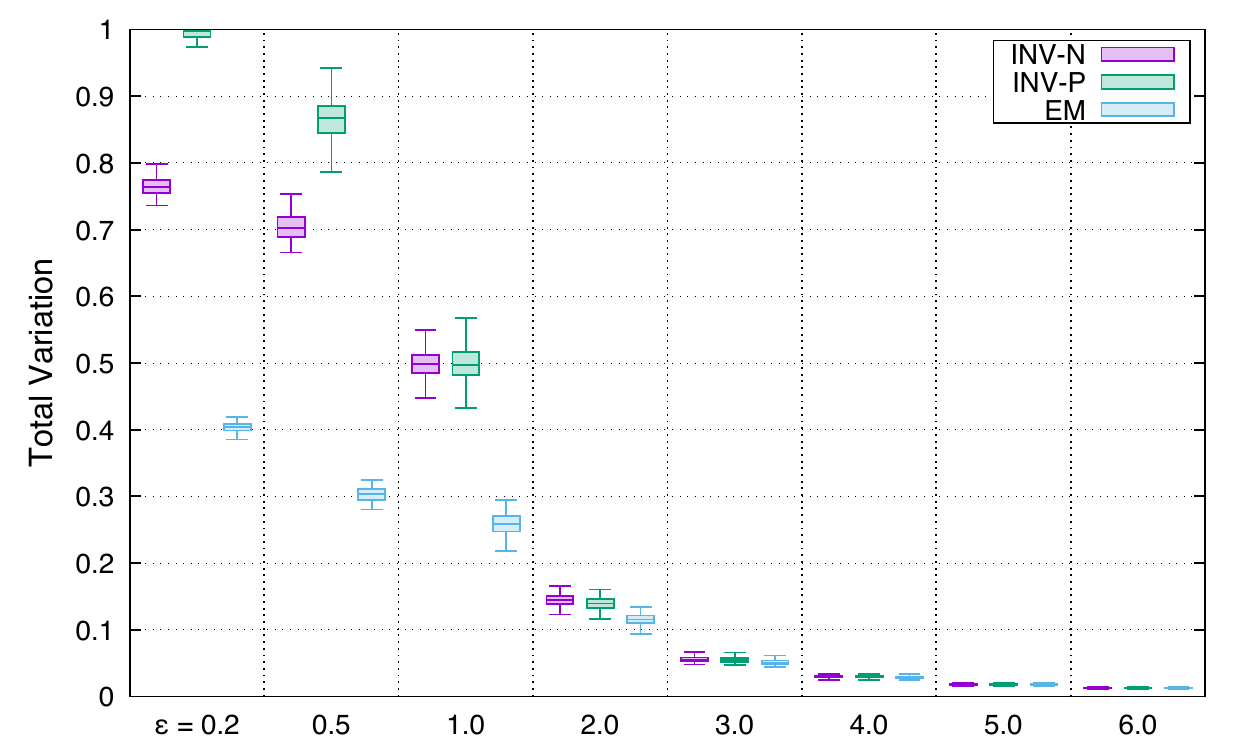}
      }
%\hspace{-10pt}
\subfigure[Earth mover's distance]{
      \label{fig:box_geom_emd}
      \includegraphics[width=0.4\textwidth]{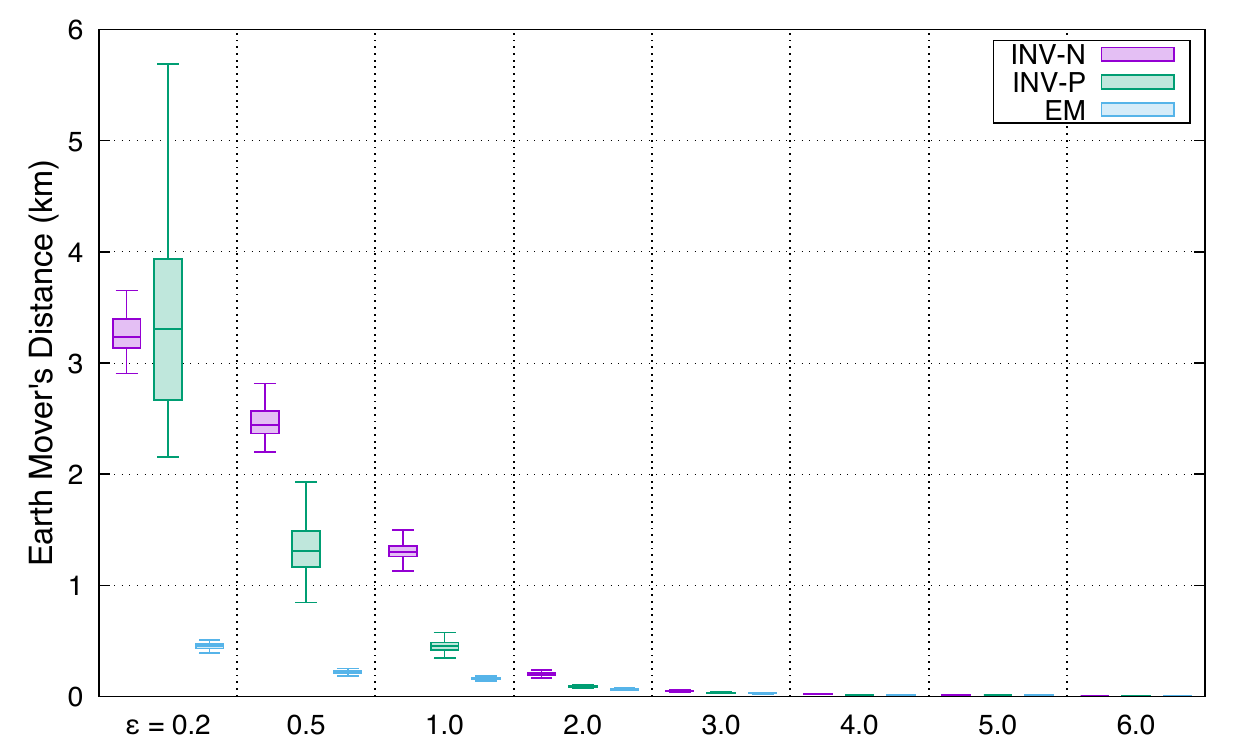}
      }
\caption{Statistical distances between the original and the estimated distributions using \emm{}, \invn{}, and \invp{}. The noisy data are produced using 
truncated planar geometric mechanisms with various values of $\epsilon$.}
\label{fig:box_estimators_geom}
\end{figure}

Fig.~\ref{fig:box_estimators_lap_exp}  shows the results for the Laplace and exponential mechanisms 
with the same range of $\epsilon$ as before, and using the TV to measure the estimation quality. 
Again we observe that \emm{} estimation method substantially outperforms \invp{} and \invn{}. We 
also observe that  the Laplace mechanism produces a better estimation  than the exponential one. This is 
because the latter introduces larger noise. This is in line with a similar result with respect to the quality of service \cite{Chatzikokolakis:17:POPETS}.
%its outstanding performance with respect to 
%the quality of service as confirmed by \cite{Chatzikokolakis:17:POPETS}.

\begin{figure}[t]
\centering 
\subfigure[Using the Laplace mechanism]{
      \label{fig:box_lap_tv}
      \includegraphics[width=0.4\textwidth]{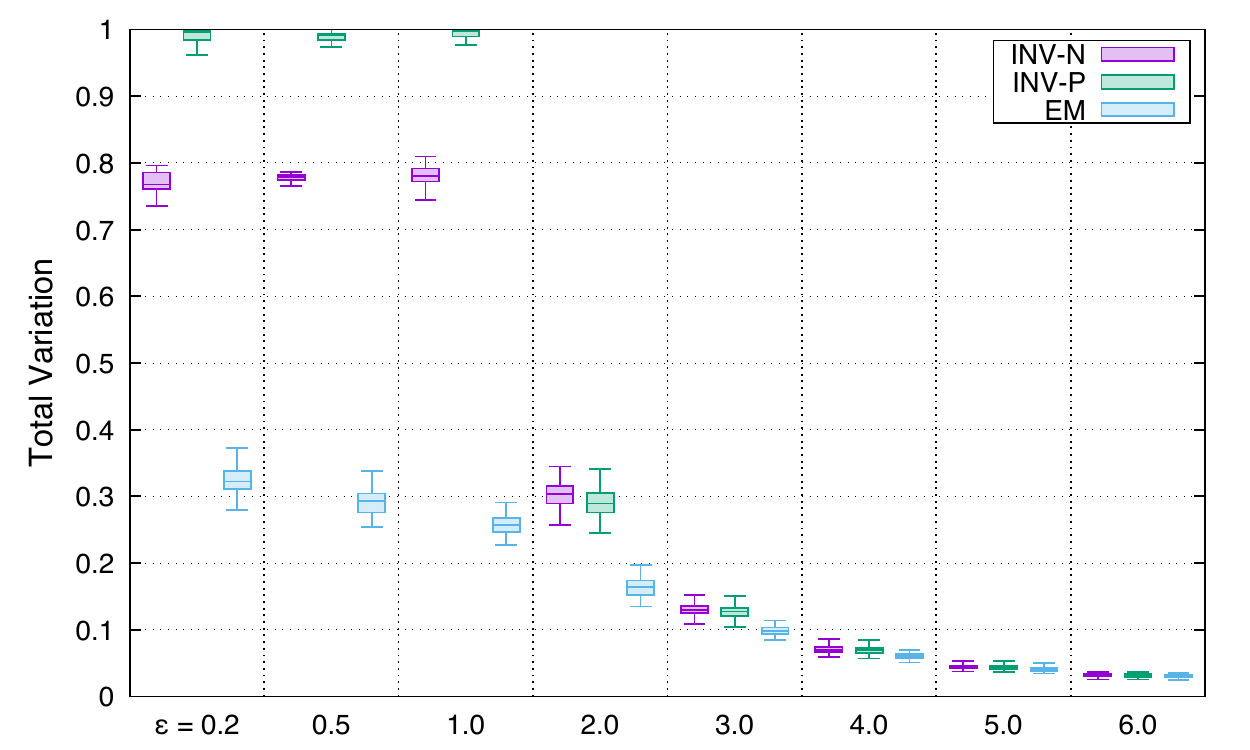}
      }
%\hspace{-10pt}
\subfigure[Using the exponential mechanism]{
      \label{fig:box_exp_emd}
      \includegraphics[width=0.4\textwidth]{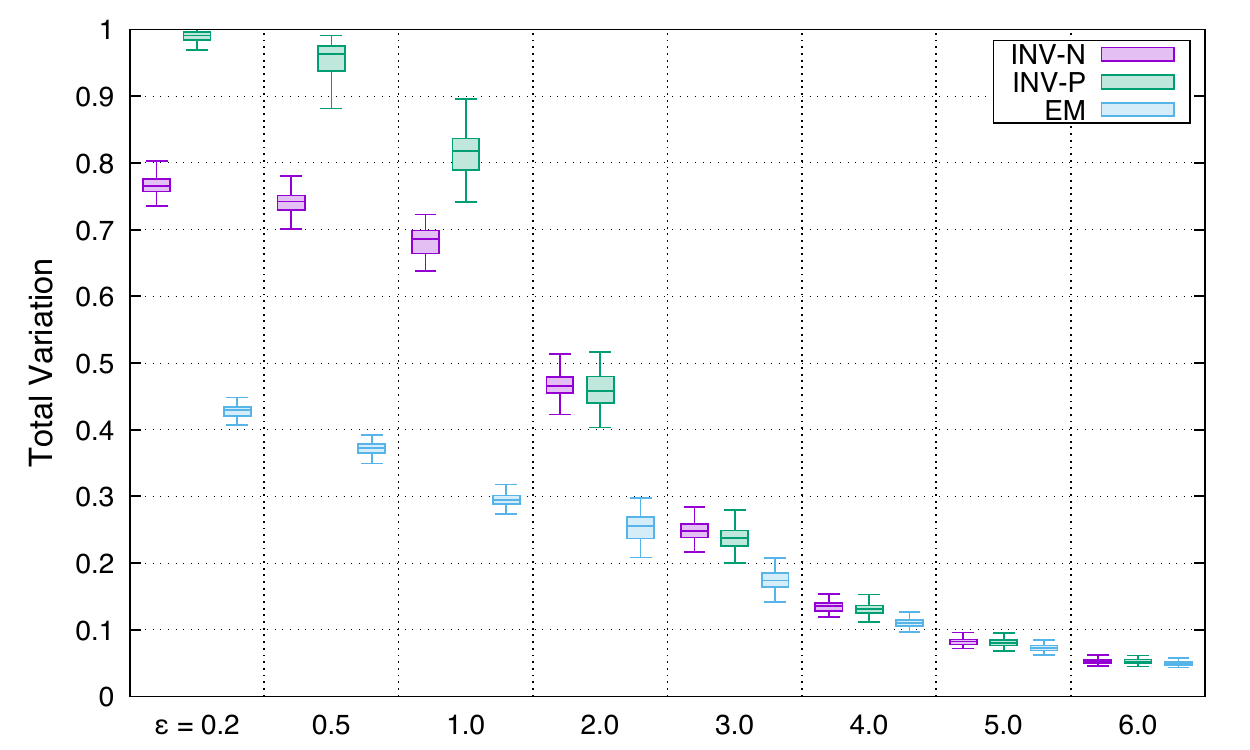}
      }
\caption{The statistical distance, measured by the total variation, between the original and estimated distributions using \emm{}, \invn{}, and \invp{}. 
The obfuscation is performed using (a) the Laplace mechanism, and (b) exponential mechanism.}
\label{fig:box_estimators_lap_exp}
\end{figure}

\subsection{Estimating distributions obfuscated by $k$-RR mechanisms}
% the following has been moved to an earlier section.
%
%The $k$-RR was originally introduced by Warner \cite{Warner:65:jastat:08:BOOK} for binary alphabets ($k$=2). Then it has been extended 
%by \cite{Kairouz:16:JMLR} to arbitrary $k$-size alphabet. This mechanism obfuscates every $x \in \calx$ to $y \in \calx$ 
%according to the following distribution. 
%\[
%P(y|x) = \frac{1}{k-1+ e^\epsilon} 
%\left\{ 
%\begin{array}{l l}
%e^\epsilon & \mbox{if $y=x$}\\
%1                & \mbox{if $y\neq x$} 
%\end{array} 
%\right.
%\] 
%It is clear that $k$-RR satisfies $\epsilon$-local differential privacy. Furthermore it is proved by \cite{Kairouz:16:JMLR} to be 
%optimal under this privacy notion for a range of statistical utility metrics, e.g. Total Variation and KL divergence. 
Now we use the $k$-RR mechanism described in Section \ref{sec:uniqueness:krr} to obfuscate the real data, and 
compare the performance of the three estimation methods on the grid of San Francisco (Fig.~\ref{fig:sf_map}).  
We apply the mechanism using various values of $\epsilon$ between $1.0$ and $10$. The results are shown in 
Fig.~\ref{fig:box_estimators_krr} in which we plot the TV and EMD distances for the three 
estimators, at every value of $\epsilon$. Unlike the cases of geometric, Laplace, and exponential mechanisms, 
we observe that the differences in the performance of estimators are not substantial. Yet, with respect to the TV distance, the \emm{} 
outperforms both \invp{} and \invn{} at all the considered values of $\epsilon$. With respect to the EMD, \emm{} and \invp{}  
have almost the same quality and superior to the \invn{}. 
\begin{figure}[t]
\centering 
\subfigure[Using the total variation distance]{
      \label{fig:box_krr_tv}
      \includegraphics[width=0.4\textwidth]{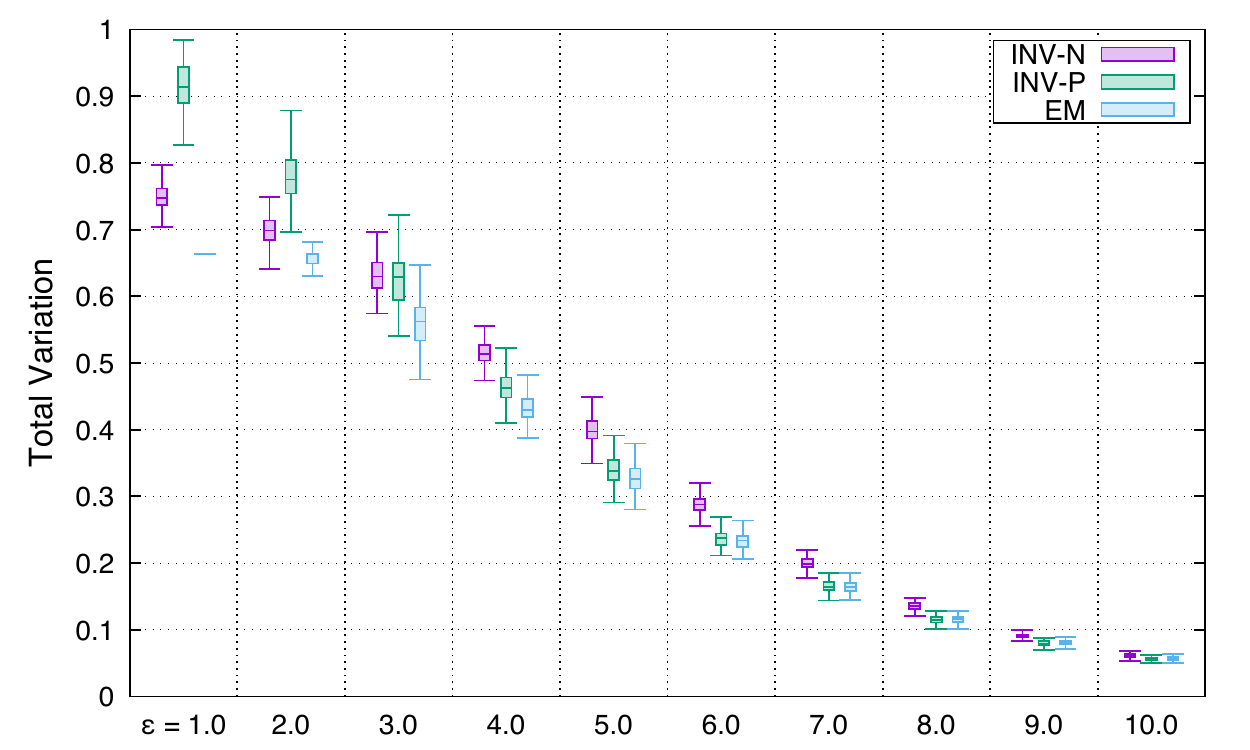}
      }
%\hspace{-10pt}
\subfigure[Using the earth mover's distance]{
      \label{fig:box_krr_emd}
      \includegraphics[width=0.4\textwidth]{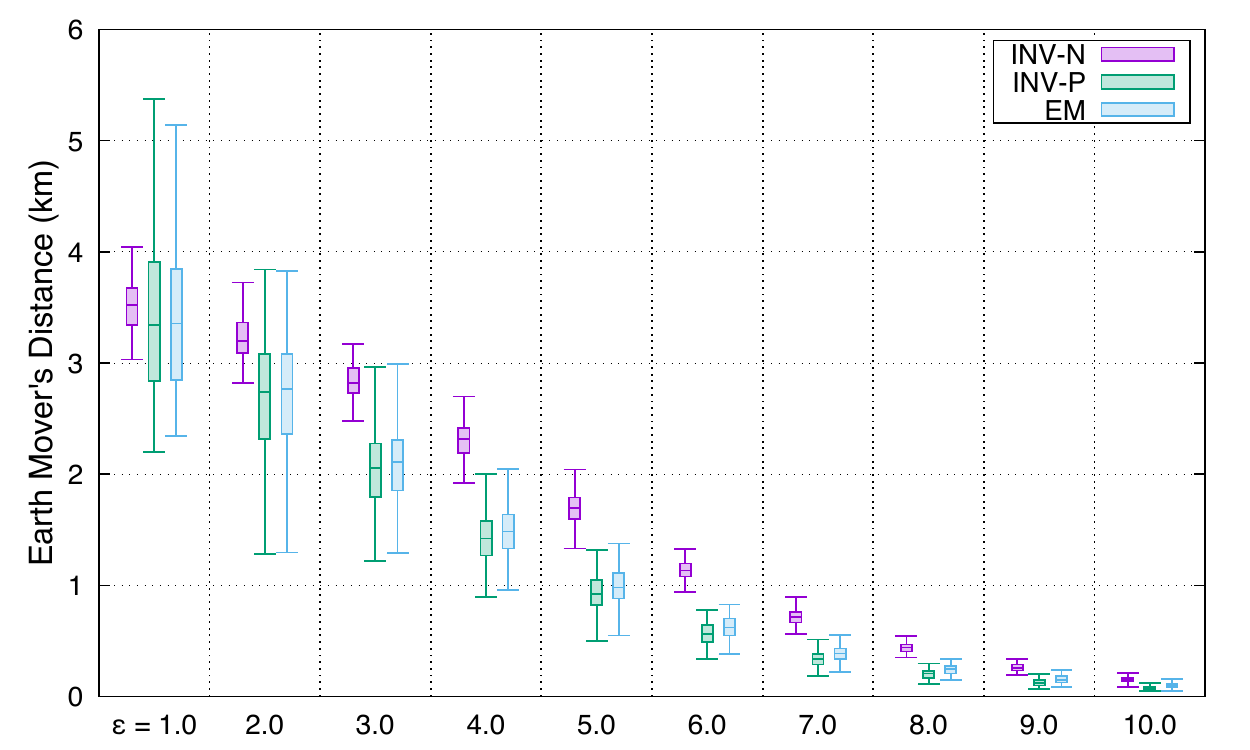}
      }
\caption{The statistical distance (measured by the total variation in (a), and by the earth mover's distance in (b)) 
between the real and estimated distributions using the indicated methods. A $k$-RR mechanism 
with indicated values of $\epsilon$ is used for obfuscating the user data.}
\label{fig:box_estimators_krr}
\end{figure}

\subsection{Estimating distributions obfuscated by \textsc{Rappor}}
\textsc{Rappor} is a mechanism built on the idea of randomized response to allow collecting statistics from 
end-users with differential privacy guarantees \cite{Erlingsson:14:CCS}. 
The basic form of this mechanism, called Basic One-Time \textsc{Rappor} maps the space of input values $\calx$ 
to a space $\calz$ of size $2^{|\calx|}$. More precisely every user datum $x \in \calx$ is encoded in a bit array $B$ of size $|\calx|$
in which only one bit $b_x$ that uniquely corresponds to $x$ is set to $1$, and other bits are set to $0$. 
Then every bit $b$ of $B$ is obfuscated independently to yield the same value $b$ with probability 
$e^{\epsilon/2}/(1+e^{\epsilon/2})$, and $1-b$ with probability $1/(1+e^{\epsilon/2})$. This obfuscation yields a noisy 
bit vector $B'$ that is reported to the server. It can be shown that \textsc{Rappor} satisfies $\epsilon$-local differential privacy.

In order to see the estimation quality of \emm{} with \textsc{Rappor}, we define $\calx= \{0,1, \dots, 9\}$, which is therefore 
mapped to $\calz$ of size $2^{10}$. We sample $10^5$ inputs (the real data) from two distributions: a binomial with $p=0.5$, 
and a uniform distribution on $\{3,4,5,6\}$ with probability $0$ to other elements. In each one of these cases, we 
obfuscate the inputs using a \textsc{Rappor} with $\epsilon=0.5$, and then apply \emm{} to estimate the original distribution. 
Fig.~\ref{fig:rappor_on_line_em} shows the results of this experiment.
\begin{figure}[t]
\centering 
\subfigure[Binomial distribution]{
      \label{fig:binomial_rappor_em}
      \includegraphics[width=0.22\textwidth]{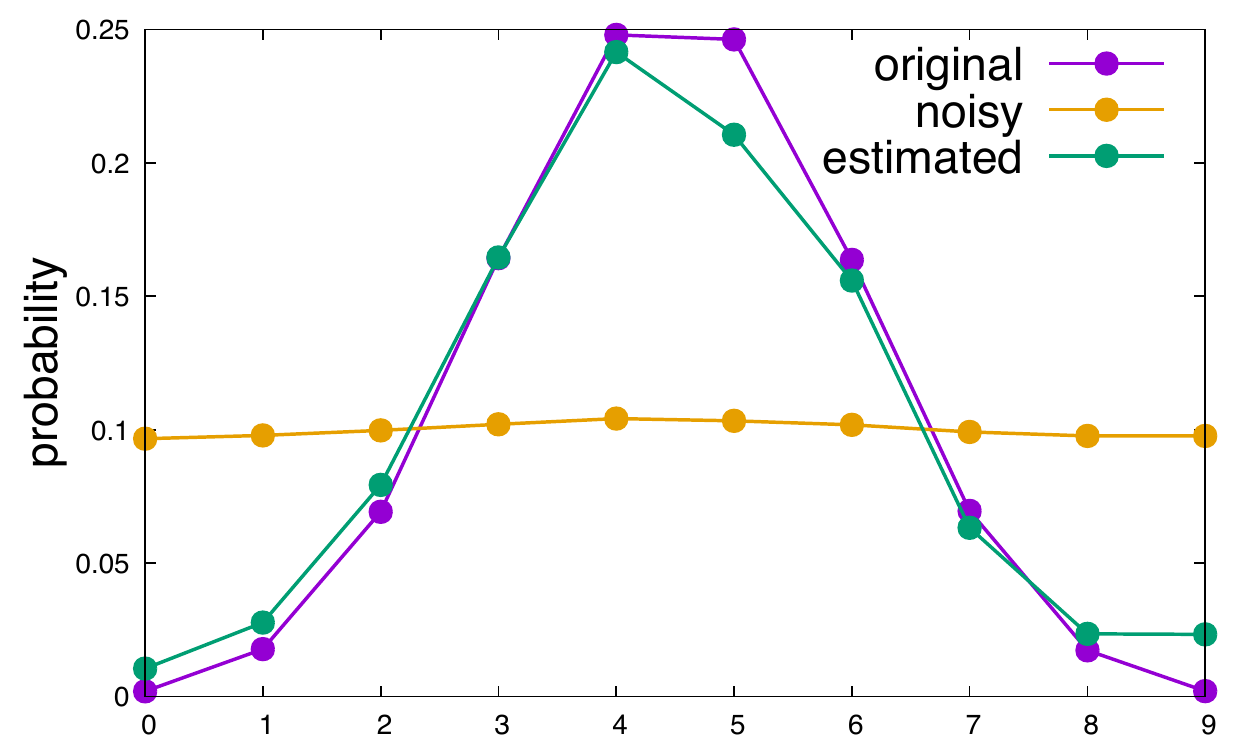}
      }
%\hspace{-10pt}
\subfigure[Uniform on $\{3,4,5,6\}$]{
      \label{fig:uniform_rappor_em}
      \includegraphics[width=0.22\textwidth]{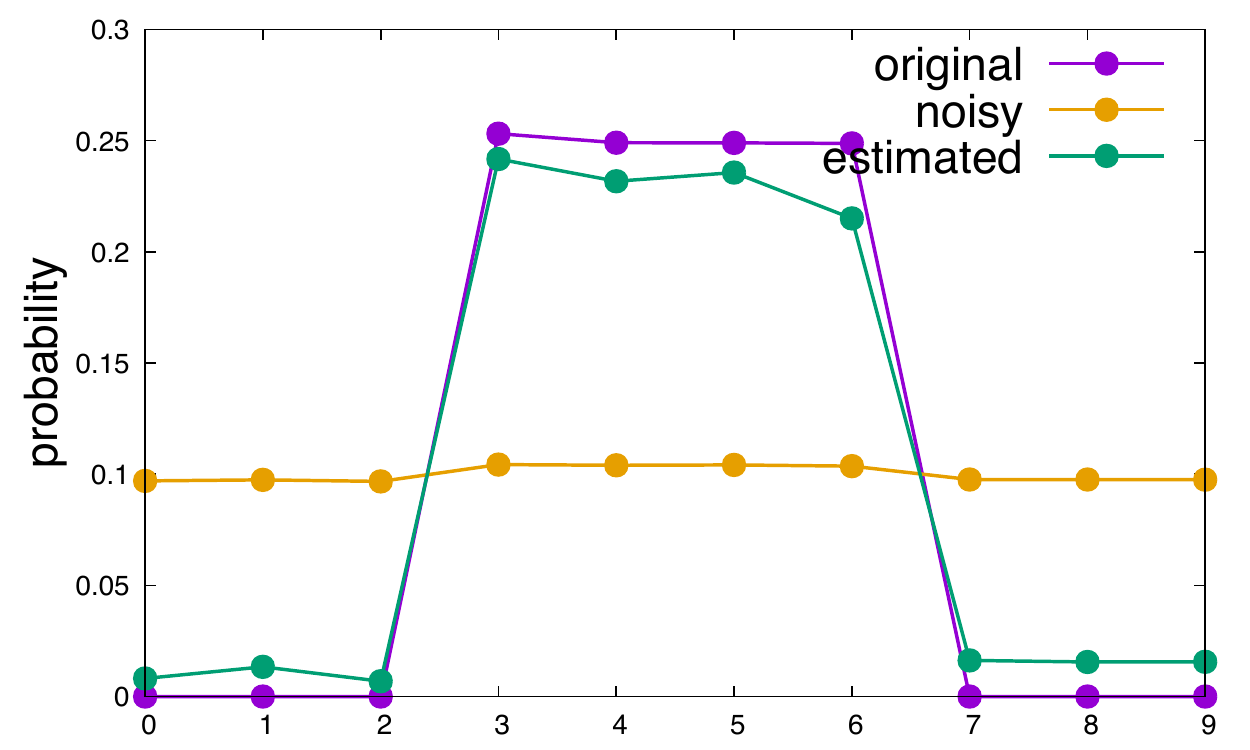}
      }
\caption{Using \emm{} to estimate the original distributions on $\{0,1,\dots, 9\}$ from noisy data produced by a basic one-time 
\textsc{Rappor} mechanism with $10^5$ user inputs and $\epsilon=0.5$. 
In (a), the real distribution is binomial. In (b), the real distribution is 
uniform on $\{3,4,5,6\}$, and assigns probability zero outside this range.}
\label{fig:rappor_on_line_em}
\end{figure}

Based on the fact that \textsc{Rappor} uses the binary randomized response obfuscation, 
the authors of \cite{Kairouz:16:ICML} adapted \invn{} and \invp{} to \textsc{Rappor}. 
We compare the performance of these methods to that 
of \emm{} by running the above experiment using the three estimators for $100$ times. In every run we evaluate the 
total variation (to the original distribution) for every method. Fig.~\ref{fig:box_estimators_rappor} shows the results of 
this procedure for a range of $\epsilon$.  
\begin{figure}[t]
\centering 
\subfigure[Original data sampled from binomial distribution]{
      \label{fig:box_rappor_binomial}
      \includegraphics[width=0.4\textwidth]{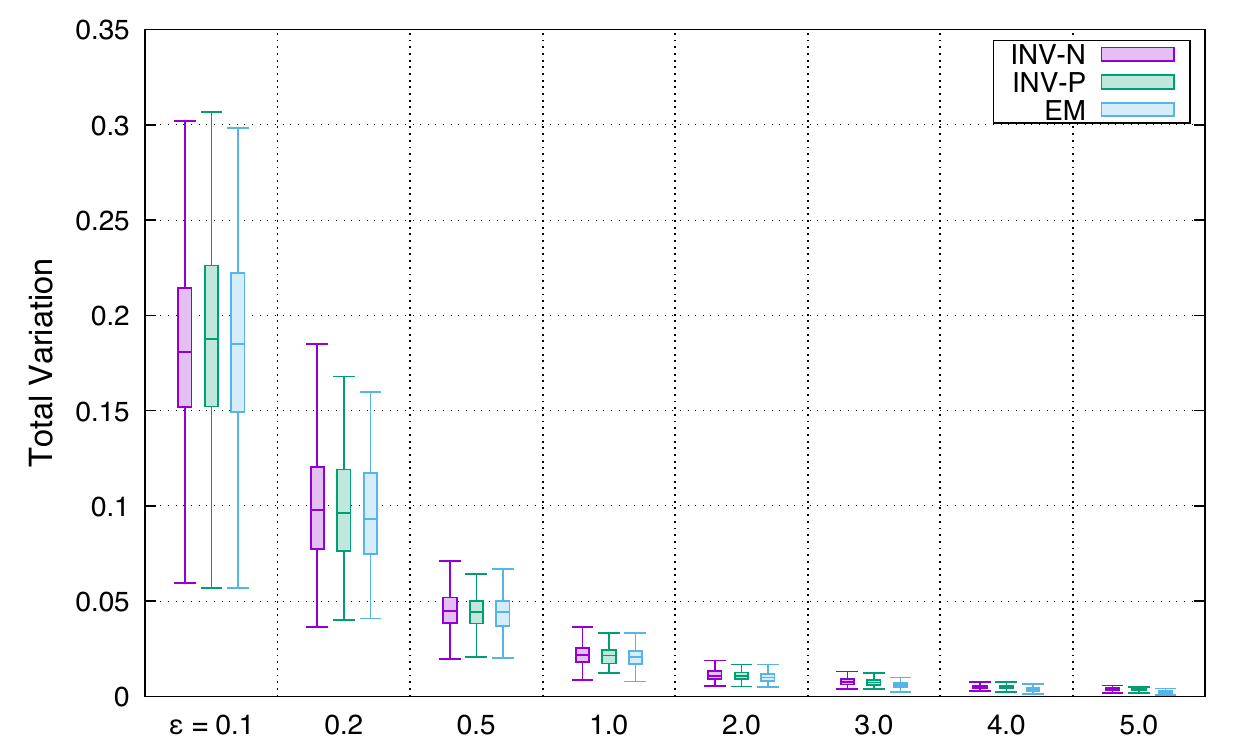}
      }
%\hspace{-10pt}
\subfigure[Original data sampled from uniform distribution]{
      \label{fig:box_rappor_uniform}
      \includegraphics[width=0.4\textwidth]{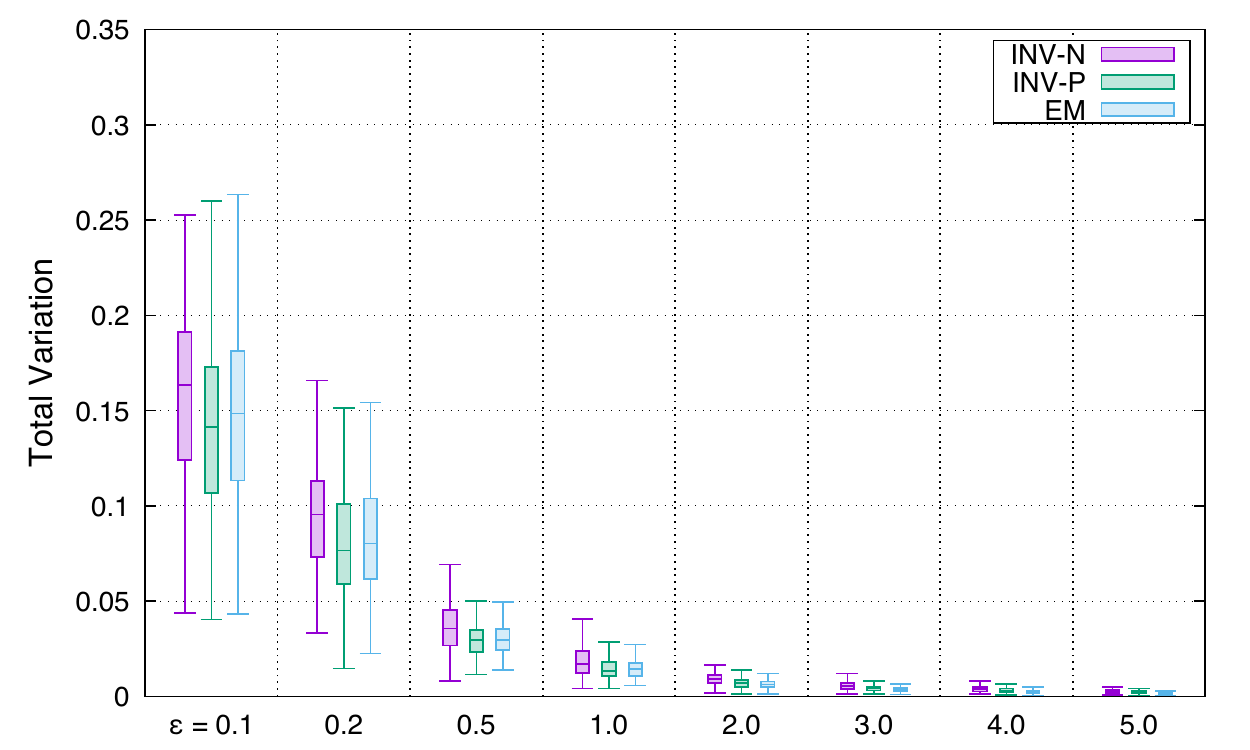}
      }
\caption{The statistical distance between the original and estimated distributions using \emm{}, \invn{}, and \invp{},  
using a \textsc{Rappor} mechanism for obfuscation. Original data of users are sampled from Binomial distribution (a), 
and from uniform distribution on $\{3,4,5,6\}$ (b).}
\label{fig:box_estimators_rappor}
\end{figure}
It can be seen that \invn{}, \invp{} and \emm{} exhibit almost the same estimation quality. This observation is anticipated 
from the similar results that we observed with the $k$-RR mechanism in Fig.~\ref{fig:box_estimators_krr} since both 
\textsc{Rappor} and $k$-RR are based on the randomized response technique.
\section{Related work}\label{sec:related_work}
The papers which are closest to our work have been already discussed in the introduction, so we will not mention them again here. 

Besides the local privacy mechanisms considered in this paper, there are many others that have been proposed, aimed at obtaining a good trade-off between privacy and precision of the statistical analyses.  Apple uses a method similar to $k$-RR to protect its users' privacy in applications such as spelling \cite{Apple:DP}.  Bassily and Smith use Random Matrix Projection  (BLH)  \cite{Bassily:15:STOC}, based on ideas from earlier works \cite{Mishra:06:SPDS,Duchi:13:FOCS}. Wang et al. \cite{Wang:17:USENIX}  have proposed the Optimized Unary
Encoding protocol, which  builds upon the Basic \textsc{Rappor}, and the  Optimized Local Hashing protocol, which is  inspired by BLH \cite{Bassily:15:STOC}.
Wang et al. \cite{Wang:16:BOOK} use both generalized random response
 and Basic \textsc{Rappor} for learning weighted histogram.

\section{Conclusion}\label{sec:conclusion}
In this paper we have analyzed the iterative Bayesian update (IBU) and we have exhibited various flaws in the underlying theory as presented in the literature. 
These flaws raise critical questions about the soundness of IBU. 
Therefore, we have provided new foundations for the IBU, and extended it to a more general local privacy model which allows, for example, the liberty of every user 
to choose his own mechanism, and to change mechanism while the data are collected. We have also studied various instances of this model.

Finally, we have compared the IBU with the matrix inversion method (INV), showing that, while for mechanisms like $k$-RR and \textsc{Rappor} the results are comparable, 
with the geo-indistinguishable  and the exponential mechanisms the results are clearly in favor of the IBU. 

%In the future we plan in the future to extend these experiments to other settings in which multiple mechanisms may be used.

As future work, we also plan to investigate   
how the choice of the   obfuscation mechanism  affects the performance of the IBU. 
This study should lay the basis for  identifying  mechanisms   which optimize the  trade-off between privacy and statistical utility.

\bibliographystyle{ieeetran}
\bibliography{short,new_refs}

\appendix
%\twocolumn

\appndexsec{The global convergence theorem}
This theorem introduced by Zangwill \cite{Zangwill:69:BOOK} describes for an iterative algorithm $M$ 
sufficient conditions that assures a sequence generated by $M$ in some space $\calc$ to converge to a 
set of points $\Gamma\in \calc$, called a solution set. Here $M$ is modeled as a set valued function mapping 
every point in $\calc$ to a set of points $M(\vtheta) \subset \calc$.    
%
%
%Since all algorithms derived from the EM framework are iterative we describe in the following the general 
%`global convergence theorem'  which specifies for an iterative algorithm $M$ 
%and an arbitrary set of points 
%$\Gamma \subset \calc$ the conditions under which the points generated by $M$ converge to $\Gamma$, 
%which is called a solution set. Here the iterative algorithm $M$ is modeled as a set-valued function which 
%maps every point $\vtheta$ in the parameter space $\calc$ to a set of points $M(\vtheta) \subset \calc$. 
%This indeed implies that any sequence of points $(\vtheta^t)_{t\in\naturals}$ generated by $M$ satisfies 
%$\vtheta^{t+1} \in  M(\vtheta^t)$. The map (algorithm) $M$ is said to be closed at $\vtheta$ if it satisfies for any 
%sequence $(\vtheta^t)_{t\in\naturals}$ converging to $\vtheta$ and any other sequence 
%$(\vtheta'^t: \vtheta'^t \in M(\vtheta^t))_{t\in\naturals}$ converging to $\vtheta'$, that $\vtheta' \in M(\vtheta)$. 
%We remark that modeling an iterative algorithm as a point-to-set map instead of a point-to-point map 
%provides generality to capture those algorithms in which the next point $\vtheta^{t+1}$ is arbitrarily 
%picked from the set $M(\vtheta^t)$ 
%\footnote{An example is the EM when there are multiple maximizers for function $Q(\vtheta | \vtheta')$, 
%and also the general-EM algorithms \cite{DLR:77:jstat}}.
%%
%Now we can write the global convergence theorem for a given map $M$ as follows. 
%
\begin{theorem}[\cite{Zangwill:69:BOOK}]
\label{thm:gct}
Consider a sequence $(\vtheta^t)_{t\in\naturals}$ satisfying $\vtheta^{t+1}\in M(\vtheta^t)$, 
where $M$ is a point-to-set map on $\calc$. Let a solution set $\Gamma \subset \calc$  
be given, and suppose that: 
(i) all points $\vtheta^t \in\calc$ are contained in a compact subset of $\calc$; 
(ii) $M$ is closed at all points of $\calc \setminus \Gamma$; 
(iii) there is a real-valued continuous function $f$ on $\calc$ such that 
\begin{align*}
&f(\vtheta') > f(\vtheta) \quad \forall  \vtheta\not\in \Gamma, \vtheta' \in M(\vtheta),  \\
&f(\vtheta') \geq f(\vtheta) \quad \forall  \vtheta\in \Gamma, \vtheta' \in M(\vtheta).
\end{align*} 
Then all the limit points of $(\vtheta^t)_{t\in\naturals}$ are in the solution set $\Gamma$ and $f(\vtheta^t)$ converges
monotonically to $f(\vtheta)$ for some $\vtheta \in \Gamma$.
\end{theorem}
%
%
%In the above theorem, the limit points of the sequence $(\vtheta^t)$, also called its accumulation points,  
%are precisely the limits of all its convergent subsequences. Note that these subsequences exist since 
%$(\vtheta^t)$ is enclosed in a compact subset of $\calc$. 
%% There is a theorem stating that every sequence in a compact set has a subsequence that converge in the same set.
%Therefore, while $(\vtheta^t)$ itself is not guaranteed to converge, 
%the theorem confirms that the limits of its convergent subsequences are in $\Gamma$, and also that the 
%function $f(.)$ running on the elements of $(\vtheta^t)$ converges to a value that corresponds 
%to a point (a parameter value) in $\Gamma$. Thus we will refer to this point as the algorithm `termination' point. 

\appndexsec{planar geometric mechanisms}
\label{sec:plan_geom}
%Let $\calx$ be a space of locations, and $\calz$ be a set of discrete observables.
%Then a mechanism satisfies $\epsilon$-geo-indistinguishability (for $\epsilon>0$) if 
%\[
%P(z | x) \leq e^{\epsilon \, \bf{d}_p(x, x')}\, P(z | x') \quad\forall x,x' \in \calx, z \in \calz 
%\]  
%where $P(z | x)$ is the probability that the mechanism reports $z$ when the real location is $x$.
%
Suppose that the planar space is discretized by an infinite grid of squared cells, where the side length
of each cell is $s$. Let $\calg$ be the set of centers of these cells. Then every 
element of $\calg$ is indexed by its coordinates $(i,j) \in \mathbb{Z}^2$. For any $\calx \subseteq \calg$, 
a \emph{planar geometric mechanism} (parametrized by $\epsilon,s$) reports a point $z\in\calg$ 
from a real location $x\in\calx$ according 
to the probability 
\begin{align}
P(z | x) &= \lambda\, e^{-\epsilon\, \distp( x , z )} \qquad x\in \calx, z \in \calg \ \label{eq:pg} \\
	        \text{where }
	        \lambda &= 1 / \sum_{(i,j) \in \mathbb{Z}^2} e^ {- \epsilon\, s \, \sqrt{i^2+j^2}},
	                \nonumber
\end{align}
and $\distp$ is the planar (i.e. Euclidean) distance. 
For $\calx$ finite we can defined a \emph{truncated} version of the above geometric mechanism. Basically it is obtained by drawing  points in $\calg$ 
according to  the above distribution and then remapping each of them to its 
nearest point in $\calx$.  

\appndexsec{proofs}
We present here the Lemmas and proofs omitted from the paper due to space constraints.

%given that the function $Q(\vtheta | \vtheta')$ satisfies mild conditions. 
%For this purpose, we first show by the following proposition a tight relationship between the derivative 
%of the log-likelihood $L(\vtheta)$ and that of the complete-data likelihood $Q(\vtheta | \vtheta')$ along any curve 
%in $\calc$.  
%\LQderivatives*
\begin{restatable}[Derivative of the log-likelihood function]{lemma}{LQderivatives}
\label{prop:LQ-derivatives}
Consider a parameter space $\calc \subset \reals^m$ and any curve $\valpha: I \to \calc$. 
Then for every $\tau' \in I$ it holds  that
$
\frac{d L(\valpha(\tau))}{d \tau} = \frac{d Q(\valpha(\tau) | \valpha(\tau'))}{d \tau} 
$ 
at $\tau=\tau'$.
\end{restatable}
\begin{proof}
By (\ref{eq:LQH}), it is clear that the statement of the proposition is equivalent to
$d H(\valpha(\tau) | \valpha(\tau')) / d \tau = 0$ at $\tau=\tau'$, where $H(.|.)$ is defined by (\ref{eq:H}). 
We start by rewriting $H(.|.)$ in a more convenient form as follows.  
\[
H( \valpha(\tau) | \valpha(\tau') ) = - \sum_{w \in \cals} \psi_w(\tau') \log \psi_w(\tau)
\]
where $\psi_w(\tau) = P(S=w | O = o ; \valpha(\tau))$ for every $w \in \cals$.  
%and $S,O$ are the random variables representing respectively hidden and observed data. 
For more convenience we will write $\psi'_w = \psi_w(\tau')$.  
Now we can evaluate the derivative of $H$ with respect to $\tau$ using the 
chain rule of derivatives as follows. 
\begin{equation}\label{eq:dHdt}
\frac{d H(\valpha(\tau) | \valpha(\tau'))}{d \tau} = \sum_{s \in \cals} \frac{\partial H}{\partial \psi_s} \frac{d \psi_s}{d \tau}. 
\end{equation}
Using the definition of $H$ and the linearity of $\partial/{\partial \psi_s}$ we have 
\[
\frac{\partial H}{\partial \psi_s} = - \sum_{w \in \cals} \frac{\partial}{\partial \psi_s}\left(\psi'_w \log \psi_w\right) = - \psi'_s/\psi_s. 
\]
By substituting the above expression in (\ref{eq:dHdt}), and observing that at $\tau=\tau'$
we have $\psi_s=\psi'_s$ for all $s \in \cals$, we obtain
\[
\frac{d H(\valpha(\tau) | \valpha(\tau'))}{d \tau} \Big |_{\tau=\tau'} = - \sum_{s \in \cals} \frac{d \psi_s}{d \tau} = 0
\]
where the last equality follows from $\sum_{s \in \cals} \psi_s(\tau) =1$. 
%\qed
\end{proof}

\EMconvergence*
\begin{proof}
We prove the theorem by showing that the conditions of Theorem \ref{thm:gct} are satisfied. 
Condition (i) is clearly satisfied. 
%since $\calc$ is compact. 
Condition (ii) is satisfied by the 
continuity of $Q(\vtheta | \vtheta')$ in both $\vtheta, \vtheta'$ as noted by \cite{Wu:83:jastat}. 
For Condition (iii) we define $f(.) = L(.)$ which is continuous, and proceed to 
prove the two parts of (iii) as follows. 
The map $\mapem$ of the EM algorithm is defined for every $\vtheta' \in \calc$ (by the M-step) 
as follows. 
\begin{equation}\label{eg:em-map}
\mapem(\vtheta') = \{ \vtheta \in \calc : Q(\vtheta |\vtheta') \geq Q(\vtheta'' |\vtheta')\,\forall \vtheta'' \in \calc\}.
\end{equation}
It is clear that the definition of $\mapem$ together with (\ref{eq:dQdL}) implies 
that $L(\vtheta) \geq L(\vtheta')$
for all $\vtheta' \in \calc$ and $\vtheta \in \mapem(\vtheta')$, hence satisfying the second part of 
Condition (iii) in Theorem \ref{thm:gct}. 
It remains to prove first part of (iii) with respect to $\mapem, \xsolnset$. 
Consider any $\vtheta' \not\in \xsolnset$. Then by Definition \ref{def:ext_solution_set} of $\xsolnset$, 
there must be a curve $\valpha:I \to \calc$ such that $\vtheta' = \valpha(\tau')$ for some $\tau' \in I$ 
and $\vtheta'$ is neither stationary nor a local maximizer along $\valpha$. This means 
\begin{align}
&| d L(\valpha(\tau)) / d \tau |_{\tau=\tau'} = c >0, \label{eq:non_s} \\
&\forall \delta>0 \, \exists \tau'' \in I :0<|\tau'' - \tau'|<\delta, \nonumber\\
&\qquad\qquad\qquad\qquad L(\valpha(\tau''))>L(\valpha(\tau')) \label{eq:non_m}. 
\end{align}
Using the definition of derivatives and that $|\lim_{\tau \to \tau'} f(\tau)| = \lim_{\tau \to \tau'} |f(\tau)| $
for any function $f$, 
%  | lim f(t) | = lim | f(t) |
% This can be proved by the reverse triangle inequality (RTI) and squeeze (S) theorem:
%  i.e. let lim f(t) =f. Then by the RTI 0 <= | |f(t)|−|f| |<=| f(t)−f |. 
%  Since 0-->0 and | f(t)−f |-->0 then by S theorem, | |f(t)|−|f| | --> 0
%  which means   lim |f(t)|  = |f| = | lim f(t) |. 
we imply from (\ref{eq:non_s}) that 
\begin{equation}\label{eq:non_s2}
\lim_{\tau \to \tau'} \Big| \frac{L(\valpha(\tau)) - L(\valpha(\tau')) }{\tau - \tau'} \Big| = c.
\end{equation}
It also follows from Proposition \ref{prop:LQ-derivatives} and the definition of derivatives that 

\begin{align}\label{eq:non_m2}
\lim_{\tau \to \tau'} \Big( \frac{Q(\valpha(\tau)|\valpha(\tau')) - Q(\valpha(\tau')|\valpha(\tau'))}{ \tau - \tau' } &- \nonumber\\ 
\frac{L(\valpha(\tau)) - L(\valpha(\tau')) }{ \tau - \tau' } \Big) &= 0.
\end{align}

% lim f(t) = lim g(t) implies lim ( f(t) - g(t) ) =0 (can be proved using the eps/delta definition of limits)
%
Choose any $\epsilon$ such that $0<\epsilon<c/2$. Then by the definition of limits there 
are $\delta_\epsilon>0$ and (by (\ref{eq:non_m})) $\tau''\in I$ with $|\tau'' - \tau'|<\delta_\epsilon$ 
and $L(\valpha(\tau'')) > L(\valpha(\tau'))$ such that
\begin{align}
&- \epsilon < \frac{L(\valpha(\tau'')) - L(\valpha(\tau')) }{| \tau'' - \tau' |} - c < \epsilon, \label{eq:dL} \\ 
&- \epsilon < \frac{Q(\valpha(\tau'')|\valpha(\tau')) - Q(\valpha(\tau')|\valpha(\tau'))}{| \tau'' - \tau' |} - \nonumber\\
&\qquad\qquad\qquad\qquad \frac{L(\valpha(\tau'')) - L(\valpha(\tau')) }{| \tau'' - \tau' |} < \epsilon \label{eq:dQ}
\end{align}
Inequalities in (\ref{eq:dL}) follow from (\ref{eq:non_s2}) with the substitution 
$|L(\valpha(\tau'')) - L(\valpha(\tau'))| = L(\valpha(\tau'')) - L(\valpha(\tau'))$ since 
$L(\valpha(\tau'')) > L(\valpha(\tau'))$.
Inequalities (\ref{eq:dQ}) follow from (\ref{eq:non_m2}). 
Now from the left inequalities in both (\ref{eq:dL}) and (\ref{eq:dQ}) we obtain respectively
\begin{align*}
&(c-\epsilon)\, |\tau'' - \tau'| <  L(\valpha(\tau'')) - L(\valpha(\tau')),\\
&L(\valpha(\tau'')) - L(\valpha(\tau')) - \epsilon | \tau'' - \tau' | < \\
&\qquad\qquad\qquad\qquad Q(\valpha(\tau'')|\valpha(\tau')) - Q(\valpha(\tau')|\valpha(\tau')). 
\end{align*}
By substituting the first inequality in the second one, we obtain
\[
(c- 2\epsilon)|\tau'' - \tau'| < Q(\valpha(\tau'')|\valpha(\tau')) - Q(\valpha(\tau')|\valpha(\tau'))
\] 
which implies that $Q(\valpha(\tau'')|\valpha(\tau')) > Q(\valpha(\tau')|\valpha(\tau'))$ because $c > 2\epsilon$. 
Since $\vtheta' = \valpha(\tau')$ and $\valpha(\tau'') \in \calc$ then the above inequality together with the definition 
of $\mapem$ in (\ref{eg:em-map}) imply that for all $\vtheta \in \mapem(\vtheta')$ we have 
$Q(\vtheta | \vtheta') > Q(\vtheta'|\vtheta')$ which implies by (\ref{eq:dQdL}) that $L(\vtheta) > L(\vtheta')$.  Thus the first 
part of condition (iii) in Theorem \ref{thm:gct} is satisfied. 
%\qed 
\end{proof}

\convex*
\begin{proof}
By Definition \ref{def:ext_solution_set} of $\xsolnset$, it is clear that every global maximizer of $L(.)$ 
is in $\xsolnset$. We show in the following that every element $\vtheta\in \xsolnset$ is a global maximizer. 
Consider any $\vtheta\in \xsolnset$ and $\vphi \in \calc$, and assume (for a contradiction) that $L(\vtheta) < L(\vphi)$. 
Consider the curve $\valpha(\tau)=(1-\tau)\vtheta + \tau \vphi$ for $\tau\in [0,1]$ which is the straight line segment 
between $\vtheta, \vphi$. Note that $\vtheta= \valpha(0), \vphi=\valpha(1)$, and also $\valpha: [0,1] \to \calc$ since 
$\calc$ is convex. Now it can be seen that $\vtheta$ is not stationary along $\valpha$ since 
\begin{align*}
&\frac{d L(\valpha(\tau))}{d \tau} \big|_{\tau \to 0} = \lim_{\tau \to 0} \frac{L( (1-\tau)\vtheta+\tau \vphi ) - L(\vtheta)}{\tau}\geq\\
&\lim_{\tau \to 0} \frac{(1-\tau) L(\vtheta)+\tau L( \vphi) - L(\vtheta)}{\tau} = L(\vphi)-L(\vtheta)>0 
\end{align*}
where the first inequality above follows from the concavity of $L(.)$, and the last inequality follows from the 
assumption $L(\vtheta) < L(\vphi)$. We also show in the following that $\vtheta$ is not a local maximizer along 
$\valpha$. Consider any $0<\delta$, and choose any $\tau$ such that $0<\tau< \min\{\delta,1\}$. 
Then $\tau \in [0,1]$ and $|\tau - 0|<\delta$.  
% P.S. In the above we say $\tau \in [0,1]$ to strictly follow the definition of a local maximizer along a curve. 
% Note here that the interval I of the curve is [0,1]. Similarly we say that $|\tau - 0|<\delta$. where 0 
% corresponds to $theta$. 
Now by the concavity of $L(.)$ it holds that
\begin{align*} 
&L(\valpha(\tau)) \geq (1-\tau) L(\vtheta)+\tau L( \vphi) > \\ 
&\qquad\qquad\qquad\qquad (1-\tau) L(\vtheta)+\tau L( \vtheta)= L(\vtheta)
\end{align*}
where the second (strict) inequality follows from the assumption that $L(\vtheta) < L(\vphi)$ and that
$\tau>0$. Therefore $\vtheta$ is not a local maximizer. We conclude that $\vtheta$ is neither stationary 
nor a local maximizer along $\valpha$, contradicting that $\vtheta\in \xsolnset$.  
%\qed  
\end{proof}

\estep*
\begin{proof}
Since $\{X^i: 1\leq i \leq n \}$ are mutually independent,  we can write $Q(.|.)$ as follows
%(\ref{eq:estep}) as follows
\begin{align*}
Q(\vtheta | \vtheta^t) &=\sum_{i=1}^n \E \left[ \log P(X^i,Z^i | \vtheta) | Z^i= \vz^i ; \vtheta^t \right] \\
                                &=\sum_{i=1}^n \sum_{x\in\calx} 
                                                          P(X^i=x | Z^i= \vz^i ; \vtheta^t) \\[-5mm]
                                                           &\qquad\qquad\qquad\qquad \log P(X^i=x,Z^i=\vz^i | \vtheta)\\
                                &=\sum_{i=1}^n \sum_{x\in\calx}
                                                           P(X^i=x | Z^i= \vz^i ; \vtheta^t) \\[-5mm]
                                                           &\qquad\qquad\qquad\qquad\log \theta_x P(Z^i=\vz^i | X^i=x)\\
%                                &=\sum_{x\in\calx} \sum_{i=1}^n  
%                                                           P(X^i=x | Z^i= \vz^i ; \vtheta^t) \log \theta_x +
%                                     \sum_{x\in\calx} \sum_{i=1}^n P(X^i=x | Z^i= \vz^i ; \vtheta^t)\,\log P(Z^i=\vz^i | X^i=x)\\
                                &=\sum_{x\in\calx} \psi_x(\vz, \vtheta^t) \log \theta_x + \\[-5mm]
                                         &\qquad\qquad\sum_{x\in\calx} \sum_{i=1}^n P(X^i=x | Z^i= \vz^i ; \vtheta^t) \log g_{x i}.
\end{align*}
where the last equality follows from the definition of $\psi_x(\vz, \vtheta^t)$ and 
Definition \ref{def:probability_matrix} for $g_{x i}$. 
%\qed
\end{proof}

\mstep*
\begin{proof}
Since we require every $\vtheta$ to satisfy the constraint $\sum_{x} \theta_x=1$, we use the 
Lagrange multiplier $\lambda$ to preserve this constraint and therefore maximize the following function 
\[
\call(\vtheta, \lambda) = Q(\vtheta | \vtheta^t) + \lambda(\sum_{x\in \calx} \theta_x -1). 
\]
The above function is maximum at $\vtheta, \lambda$ that satisfy $\partial \call /\partial \theta_x =0$ for 
all $x\in \calx$, and also $\partial \call /\partial \lambda =0$. Then, by Theorem \ref{thm:estep} we obtain 
$\partial \call /\partial \theta_x = \psi_x(\vz, \vtheta^t)/ \theta_x +\lambda =0$ which yields 
$\theta_x = -\psi_x(\vz, \vtheta^t)/ \lambda$ for all $x \in \calx$. From the condition $\partial \call /\partial \lambda =0$ 
we obtain $\lambda = - \sum_{x\in\calx} \psi_x(\vz, \vtheta^t)= -n$ where the last equality follows from the definition of 
$\psi_x(\vz, \vtheta^t)$ in Theorem \ref{thm:estep}. Thus $\call(\vtheta, \lambda)$ is maximized at $\vtheta$ with 
$\theta_x = \psi_x(\vz, \vtheta^t)/ n$ for all $x \in \calx$. The proof is finally completed by substituting $\psi_x(\vz, \vtheta^t)$ 
with its definition in Theorem \ref{thm:estep}.
%\qed
\end{proof}

\compact*
\begin{proof}
From the definition of $\calc', \calc$ it is clear that $\calc' \subset \calc \subset \reals^{|\calx|}$.
Then by the Bolzano–Weierstrass theorem, $\calc'$ is compact if and only if it is 
bounded in $\reals^{|\calx|}$ and closed. It is clear that $\calc'$ is bounded since $\calc$ is. 
It is also closed, i.e. containing its limit points as follows. Every limit point $\vtheta^*$ of $\calc'$ 
is the limit of some sequence $(\vtheta^t)_{t\in\naturals}$ sampled from $\calc'\setminus \vtheta^*$, 
i.e. $\lim_{t\to\infty} \vtheta^t=\vtheta^*$. 
% Wikipedia (limit points): If X is a Fréchet–Urysohn space (which all metric spaces and first-countable spaces are), 
% then x ∈ X is a limit point of S if and only if there is a sequence of points in S \ {x} whose limit is x
Here it can be seen from the definitions of $\calc',\calc$ that 
$\theta^*_x \geq 0$ for all $x\in \calx$ and $\sum_{x\in \calx} \theta^*_x =1$. It also follows that 
$L^- \leq \lim_{t\to\infty} L(\vtheta^t) = L(\vtheta^*)$ where the last equality holds by the 
continuity of $L(.)$ on $\calc$, and therefore on $\calc'$. Thus $\vtheta^* \in \calc'$, 
hence $\calc'$ is closed.
\end{proof}

\convexityLPM*
\begin{proof}
$L(.)=\sum_{i=1}^n L_i(.)$ is concave because $L_i(.)$ is concave by (\ref{eq:L_lpm}) for all $i$.    
$\calc$ is convex. In fact for every two distributions $\vtheta, \vphi \in \calc$
and every $\tau \in [0,1]$, the distribution $\vtheta' = (1-\tau) \vtheta + \tau \vphi$ is in $\calc$ because 
$\vtheta'$ is a distribution and $L(\vtheta') \geq  (1-\tau) L(\vtheta) + \tau L(\vphi)> -\infty$ where the first 
inequality follows from the concavity of $L(.)$ and the second from $\vtheta, \vphi \in \calc$. 
\end{proof}

\algconv*
\begin{proof}
We show first that $(\vtheta^t)_{t\in\naturals}$ is contained in a compact subset $\calc'$ of $\calc$.
Let $\calc' = \{\vtheta \in \calc: L(\vtheta)\geq L(\vtheta^0) \}$. Since $L(\vtheta^0)>-\infty$ it follows 
from Lemma \ref{lemma:compact}, that
$\calc'$ is compact. Since also $L(\vtheta^t)$ is monotonically increasing by Theorem \ref{thm:mstep} 
and (\ref{eq:dQdL}), we have 
$L(\vtheta^t) \geq L(\vtheta^0)$ for all $t\in \naturals$. Therefore the generated sequence 
$(\vtheta^t)_{t\in\naturals}$ is contained in $\calc'$. 
Since $L(.)$ is continuous on $\calc$ and $Q(.|.)$ 
is continuous with respect to its two arguments (Theorem \ref{thm:estep}), it follows from 
Theorem \ref{thm:conv:em} that $L(\vtheta^t)$ converges to $L(\hat\vtheta)$ for some 
$\hat\vtheta \in \xsolnset$. Then by Lemma \ref{lemma:convexityLPM}, $L(.)$ is 
concave and $\calc$ is convex; and therefore $\xsolnset$ is exactly the set of global 
maximizers by Theorem \ref{thm:convex_space}. 
\end{proof}

\mlsingleinp*
\begin{proof}
Let $P(\vz | \vtheta)$ denote the probability of $\vz$ with respect to a distribution $\vtheta$ on $\calx$. 
Then for all distributions $\vtheta$ it holds that $P(\vz | \vtheta)$ is tightly upper-bounded as follows
\[
P(\vz | \vtheta) = \sum_{x\in \calx} \theta_x\, P(\vz | X= x) \leq \max_{x \in \calx} P(\vz | x). 
\]
Now define $M$ to be the above upper bound, i.e. $M = \max_{x\in \calx} P(\vz | x)$; then it holds 
for any distribution $\vtheta$ that
\begin{align*}
M - P(\vz | \vtheta) &= M \sum_{x \in \calx} \theta_x - \sum_{x \in \calx} \theta_x\,P(\vz | x) \\
&= \sum_{x \in \calx} (M - P(\vz | x)) \theta_x \geq 0.
\end{align*}
It follows from the above inequality that $P(\vz | \vtheta)$ attains the maximum value $M$ 
if and only if $\theta_x=0$ for all $x$ satisfying $M-P(\vz | x)>0$,  i.e. $x\not\in\calx$. 
This condition is equivalent to that given in the theorem statement. 
%\qed
\end{proof}

\mlunique*
\begin{proof}
Let $\calc$ be the set of all probability distributions $\vtheta$ on $\calx$ with 
finite log-likelihood, i.e. $L(\vtheta)>-\infty$. 
Let also $\calc' = \{\vtheta \in \calc : L(\vtheta)\geq L(\vtheta')\}$.    
Then it follows from Lemma \ref{lemma:compact} that $\calc'$ is compact. 
Since also $L(.)$ is continuous on $\calc'$ it has a global maximizer in $\calc'$ 
(by the extreme value theorem). Therefore $L(.)$ has  
a global maximizer in $\calc$. Now we show  
under the condition stated in the theorem that $L(.)$ is also strictly concave 
on $\calc$. 
Since $L(\vtheta)= \sum_{i=1}^n L_i(\vtheta)$, we can split 
this sum into two parts as follows. 
\begin{equation}\label{thm:unique:eq:4}
L(\vtheta) = \sum_{i\in\cali'} L_i(\vtheta)+ \sum_{i \in (\cali \setminus \cali')} L_i(\vtheta).
\end{equation}
Now let  $\vtheta, \vphi$ be any two nonidentical distributions on $\calx$, and consider any $0<\gamma<1$. 
Then by (\ref{eq:L_lpm})  it holds for every $i \in \cali'$ that    
\begin{flalign*}
&L_i(\gamma \vtheta + (1-\gamma) \vphi) = \log \sum_{u \in \calx} (\gamma \theta_u + (1-\gamma) \phi_u) g_{u i} \\ 
          &\qquad\qquad= \log \left( \gamma \sum_{u \in \calx} \theta_u g_{u i} + (1-\gamma)\sum_{u \in \calx} \phi_u g_{u i} \right) \\
          &\qquad\qquad\geq \gamma \log \sum_{u \in \calx} \theta_u g_{u i} +  (1-\gamma) \log \sum_{u \in \calx} \phi_u g_{u i}\\
          & \qquad\qquad = \gamma \, L_i(\vtheta) +(1-\gamma) \, L_i(\vphi)
\end{flalign*}
where the above inequality follows from the Jensen inequality and the concavity of the 
$\log(.)$. The inequality is strict equality if and only if
\begin{equation}\label{thm:unique:eq:2}
\sum_{u \in \calx} \theta_u g_{u i} = \sum_{u \in \calx} \phi_u g_{u i}.
\end{equation}
Therefore
\begin{equation}\label{thm:unique:eq:3}
\sum_{i\in\cali'} L_i(\gamma \vtheta + (1-\gamma) \vphi) \geq \gamma \sum_{i\in\cali'} \, L_i(\vtheta) +(1-\gamma) \, \sum_{i\in\cali'} L_i(\vphi)
\end{equation}
where the above inequality is equality if and only if (\ref{thm:unique:eq:2}) is satisfied for all $i\in \cali'$, 
i.e. $(\vtheta-\vphi)\, \mg(\cali') = \vzeros$. 
Therefore if this condition is never satisfied for every non-identical $\vtheta, \vphi$
%
%Since $\sum_{x\in \calx} (\theta_x-\phi_x)=0$ 
%and $(\vtheta-\vphi)\neq \vzeros$ because $\vtheta,\vphi$ are nonidentical distributions, then by 
%the assumption that the theorem conditions hold, $(\vtheta-\vphi)\, \mg(\cali') = \vzeros$ is never satisfied. 
%This makes 
%
the inequality in (\ref{thm:unique:eq:3}) strict, which implies by (\ref{thm:unique:eq:4}) that 
$L(.)$ is strictly concave on $\calc$ which is convex (Lemma \ref{lemma:convexityLPM}). Therefore 
the global maximizer of $L(.)$ on $\calc$ is unique.  
%\qed  
\end{proof}

\mluniquekrr*
\begin{proof}
Let $\calx =\{x_1,x_2, \dots, x_m\}$ be the values reported by the $k$-RR. We define the set $\cali'$ 
to include for every $x\in \calx$ an index (in the sequence of outputs) at which the value $x$ is observed. 
Then both the rows and columns of $\mg(\cali')$ correspond to the $m$ elements of $\calx$. 
Let $a = e^\epsilon/(1-k+ e^\epsilon)$, $b = 1/(1-k+ e^\epsilon)$, and $C_i$ be the $i$-th column 
of $\mg(\cali')$. Then by (\ref{eq:krr})
\[
C_i = [b,b,\dots,b, a, b,b,\dots, b]^T
\]
in which the $i$-th entry is $a$ and every other entry is $b$. 
%we can write $\mg(\cali')$ as 
%\[
%\mg(\cali') = \begin{bmatrix}
%    a           & b       &\dots        & b  \\
%    b           & a       &\dots        & b  \\
%    b           & b       &\dots        & b    \\
%    \vdots    & \vdots                 & \vdots  \\
%    b           & b        &\dots       & a
%\end{bmatrix}.
%\]
Let $\vones$ be the column vector having all entries equal to $1$. 
We first show that any row vector $\vv$ satisfying $\vv\,\mg(\cali') = \vzeros$, 
$\vv \vones = 0$ must be equal to $\vzeros$. These constraints imply for every $i=1,2,\dots,m$ that 
$\vv \, (C_i - b \vones) = \vzeros$ which means that $v_i=0$ for every $i$, i.e. $\vv=\vzeros$.  
We know that any two distributions $\vtheta, \vphi$ satisfy $(\vtheta- \vphi) \vones = 0$. Therefore if they 
also satisfy $(\vtheta - \vphi) \,\mg(\cali')=\vzeros$ they must satisfy $\vtheta - \vphi= \vzeros$, 
i.e. must be identical, which implies by 
Theorem \ref{thm:mlunique} that the MLE on $\calx$ is unique. 
% remark: the above set \calx is in fact the likely set, which should be denoted as \calx', if we define likely sets. 
%It is easy to see from (\ref{eq:krr})
%that the elements in $\calx\setminus\calx'$ are unlikely (Definition \ref{def:unlikely}). Then by 
%Theorem \ref{thm:domain_restriction} there is also a unique MLE on $\calx$. 
%
% ehab: The following proves the above corollary even with multiple k-RR mechanisms (we show that in a future paper)
%
% Here we show that the likelihood maximizer is unique. 
% The latter condition is useful to prove that Krr mechanism has 
% a unique maximizer over the set of observations: The equation is a system of linear equations 
% that can be solved for unknown \vv. For every x \in X, there is a column (equation) c_x in the 
% matrix K with c_x(x) = a and c_x(x')=b (fixed) all x' \neq x. Therefore subtracting the ones column 
% in K scaled by b from the column c_x yields that (a-b) v_x = 0. which means that v_x must be 0.      
% Note that in kRR the set of observations is exactly the likely set. Therefore the maximizer 
% on the entire set X is also unique. Note that (a,b) are not necessarily the same for all columns 
% therefore there is a unique maximizer even that every observation is revealed using a different 
% kRR mechanism.  
\end{proof}

\mluniquegeometric*
\begin{proof}
Let $\calx =\{x_1,x_2, \dots, x_m\}$ be the set of reported values. We define the set $\cali'$ 
to contain for every $x\in \calx$ an index (in the sequence of outputs) at which the value $x$ is observed. 
Then both the rows and columns of $\mg(\cali')$ correspond to the elements of $\calx$. We assume without 
loss of generality that $x_1<x_2<\dots<x_m$. Then using (\ref{eq:geometric}), and letting $\alpha= e^{-\epsilon}$
we can write $\mg(\cali')$ as 
\[
\mg(\cali') =c\, \begin{bmatrix}
    1                                     & \alpha^{x_2 - x_1}        & \alpha^{x_3 - x_1}             &\dots        & \alpha^{x_m - x_1}  \\
    \alpha^{x_2 - x_1}          & 1                                   & \alpha^{x_3 - x_2}             &\dots        & \alpha^{x_m - x_2}  \\
    \alpha^{x_3 - x_1}          & \alpha^{x_3 - x_2}         & 1                                       &\dots        & \alpha^{x_m - x_3}    \\
%    \alpha^{x_4 - x_1}          & \alpha^{x_4 - x_2}         & \alpha^{x_4 - x_3}            &\dots        & \alpha^{x_m - x_4}    \\
    \vdots                             & \vdots                            & \vdots                               &\ddots       &\vdots                       \\
    \alpha^{x_m - x_1}         & \alpha^{x_m - x_2}        & \alpha^{x_m - x_3}           &\dots         & 1
\end{bmatrix}.
\]
We show that any vector $\vv$ satisfying $\vv\,\mg(\cali') = \vzeros$ must be zero. This is a system 
of $m$ linear equations in which the $i$-th one corresponds to the $i$-th column $C_i$ of the above matrix  
as  $\vv \, C_i = 0$. It is easy to see that $\vv \, (C_1 - \alpha^{(x_2-x_1)} C_2)= 0$ yields $v_1=0$. Using 
 this together with $\vv \, (C_2 - \alpha^{(x_3-x_2)} C_3) = 0$ yields $v_2=0$. Repeating this procedure
 inductively on every successive two columns yields that $v_i=0$ for all $i \leq m$. Therefore any two distributions 
 $\vtheta, \vphi$ satisfying $(\vtheta - \vphi) \,\mg(\cali')=\vzeros$ must be identical, which implies by 
 Theorem \ref{thm:mlunique} that the MLE on $\calx$ is unique. 
% Note that this proof works on linear geometric mechanism but not on its planar version (subtracting 
% scaled columns does not yield the above result in this case). Note also (unlike kRR mechanism),
% the maximizer on the observations is not necessarily a maximizer on the entire space X since 
% the likely elements are not restricted to the observations, but extends over the convex hull of the 
% observations.  
 \end{proof}

\mluniquetgeometric*
\begin{proof}
The proof is similar to that of Corollary \ref{thm:mlunique-geometric} with the only difference is that every column 
in $\mg(\cali')$ is scaled by $c_{x_i}/c$ where $x_i$ for $i=1,2,\dots, m$ are the observed values, and $c_{x_i}$ 
is given by (\ref{eq:tgeom-cz}).
\end{proof}

\end{document}